%% file: factor-mod-p4.tex
\documentclass[letter,12pt,dvipsnames]{article}
\newcommand{\ignore}[1]{}

\usepackage[x11names, rgb]{xcolor}
\usepackage[utf8]{inputenc}
\usepackage{tikz}
\usetikzlibrary{snakes,arrows,shapes}
\usepackage{fullpage}
\usepackage[colorlinks=false, pagebackref]{hyperref}  
\usepackage{microtype}%if unwanted, comment out or use option "draft"
\usepackage{textcomp}
\usepackage{amsmath, amsfonts, amssymb, amsthm}
\usepackage{mathtools, mathdots}
\usepackage{mhsetup}
\usepackage{algorithm, algpseudocode, algorithmicx, float}
\usepackage{enumerate}
\usepackage{comment}

\newtheorem{theorem}{Theorem}

\newtheorem{lemma}[theorem]{Lemma}
\newtheorem{definition}[theorem]{Definition}

\newtheorem{remark}{Remark}

\newtheorem{claim}[theorem]{Claim}

\newcommand{\claimproof}[2]%
{\noindent{\em Proof of Claim \ref{#1}.}
#2\hspace*{\fill}$\Box$~~~~\vspace{3.5mm} }

\newcommand{\F}{\mathbb{F}}
\newcommand{\N}{\mathbb{N}}

\newcommand{\Z}{\mathbb{Z}}

\newcommand{\val}{\text{val}}
\newcommand{\gen}[1]{\langle #1 \rangle}

\makeatletter
\newenvironment{breakablealgorithm}
  {% \begin{breakablealgorithm}
   \begin{center}
     \refstepcounter{algorithm}% New algorithm
     \hrule height.8pt depth0pt \kern2pt% \@fs@pre for \@fs@ruled
     \renewcommand{\caption}[2][\relax]{% Make a new \caption
       {\raggedright\textbf{\ALG@name~\thealgorithm} ##2\par}%
       \ifx\relax##1\relax % #1 is \relax
         \addcontentsline{loa}{algorithm}{\protect\numberline{\thealgorithm}##2}%
       \else % #1 is not \relax
         \addcontentsline{loa}{algorithm}{\protect\numberline{\thealgorithm}##1}%
       \fi
       \kern2pt\hrule\kern2pt
     }
  }{% \end{breakablealgorithm}
     \kern2pt\hrule\relax% \@fs@post for \@fs@ruled
   \end{center}
  }
\makeatother

\makeatletter
\newcommand{\algmargin}{\the\ALG@thistlm}
\makeatother
\newlength{\whilewidth}
\algdef{SE}[parWHILE]{parWhile}{EndparWhile}[1]
  {\parbox[t]{\dimexpr\linewidth-\algmargin}{%
     \hangindent\whilewidth\strut\algorithmicwhile\ #1\ \algorithmicdo\strut}}{\algorithmicend\ \algorithmicwhile}%
\algnewcommand{\parState}[1]{\State%
  \parbox[t]{\dimexpr\linewidth-\algmargin}{\strut #1\strut}}

\begin{document}

\pagenumbering{gobble}% Remove page numbers (and reset to 1)

%\begin{titlepage}

\title{\bf Efficiently factoring polynomials modulo $p^4$}

\author{
Ashish Dwivedi \thanks{CSE, Indian Institute of Technology, Kanpur, \texttt{ashish@cse.iitk.ac.in} }
\and
Rajat Mittal \thanks{CSE, Indian Institute of Technology, Kanpur, \texttt{rmittal@cse.iitk.ac.in} }
\and
Nitin Saxena \thanks{CSE, Indian Institute of Technology, Kanpur, \texttt{nitin@cse.iitk.ac.in} }
}

\date{}
\maketitle

\begin{abstract}
Polynomial factoring has famous practical algorithms over fields-- finite, rational \& $p$-adic. However, modulo prime powers it gets hard as there is non-unique factorization and a combinatorial blowup ensues. For example, $x^2+p \bmod p^2$ is irreducible, but $x^2+px \bmod p^2$ has exponentially many factors! We present the first randomized poly($\deg f, \log p$) time algorithm to factor a given univariate integral $f(x)$ modulo $p^k$, for a prime $p$ and $k \leq 4$. Thus, we solve the open question of factoring modulo $p^3$ posed in (Sircana, ISSAC'17). 

Our method reduces the general problem of factoring $f(x) \bmod p^k$ to that of {\em root finding} in a related polynomial $E(y) \bmod\langle p^k, \varphi(x)^\ell \rangle$ for some irreducible $\varphi \bmod p$. We could efficiently solve the latter for $k\le4$, by incrementally transforming $E(y)$. Moreover, we discover an efficient and strong generalization of Hensel lifting to lift factors of $f(x) \bmod p$ to those  $\bmod\ p^4$ (if possible). This was previously unknown, as the case of repeated factors of $f(x) \bmod p$ forbids classical Hensel lifting.
\end{abstract}

\vspace{-.30mm}
\noindent
{\bf 2012 ACM CCS concept:} Theory of computation-- Algebraic complexity theory,  Problems, reductions and completeness; Computing methodologies-- Algebraic algorithms, Hybrid symbolic-numeric methods; Mathematics of computing-- Number-theoretic computations.

\vspace{-.45mm}
\noindent
{\bf Keywords:} efficient, randomized, factor, local ring, prime-power, Hensel lift, roots, p-adic.  % mandatory: Please provide 1-5 keywords

%\end{titlepage}

\pagenumbering{arabic}% Arabic page numbers (and reset to 1)

\vspace{-1mm}

\input{intro.tex}

\input{prelim.tex}

\input{sec_3-1.tex}

\input{sec_3-2.tex}

\input{sec_3-3.tex}

\input{sec_3-4.tex}

\input{sec_3-5.tex}

%\input{futurework.tex}

\section{Conclusion}
\vspace{-1mm}

The study of \cite{von1998factoring, von1996factorization} sheds some light on the behaviour of the factoring problem for integral polynomials modulo prime powers. 
It shows that for ``large" $k$ the problem is similar to the factorization over $p$-adic fields (already solved efficiently by \cite{cantor2000factoring}). 
But, for ``small" $k$ the problem seems hard to solve in polynomial time. We do not even know a practical algorithm.

This motivated us to study the case of constant $k$, with the hope that this will help us invent new tools. 
In this direction, we make significant progress by giving a unified method to factor $f\bmod p^k$ for $k\leq 4$. 
We also generalize Hensel lifting for $k\leq 4$, by giving all possible lifts of a factor of $f\bmod p$, in the classically hard case of $f\bmod p$ being a power of an irreducible.

We give a general framework (for any $k$) to work on, by reducing the factoring in a big ring to root-finding in a smaller ring. 
We leave it open whether we can factor $f\bmod p^5$, and beyond, within this framework.

We also leave it open, to efficiently get all the solutions of a {\em bivariate} equation, in $\Z/\langle p^k\rangle$ or $\F_p[x]/\langle \varphi^k\rangle$, in a  compact representation. 
Surprisingly, we know how to achieve this for univariate polynomials~\cite{berthomieu2013polynomial}. 
This, combined with our work, will probably give factoring mod $p^k$, for any $k$.

%\smallskip
\noindent
{\bf Acknowledgements. } 
We thank Vishwas Bhargava for introducing us to the open problem of factoring $f \bmod p^3$. N.S.~thanks the funding support from DST (DST/SJF/MSA-01/2013-14). 
R.M.~would like to thank support from DST through grant DST/INSPIRE/04/2014/001799. 
%We thank Manindra Agrawal, Sumanta Ghosh, Partha Mukhopadhyay, Thomas Thierauf and Nikhil Balaji for the discussions.
\medskip

\bibliographystyle{alpha}
\bibliography{bibliography}

\appendix

\input{appendix.tex}

\end{document}

%% file: intro.tex
\section{Introduction}
\label{sec-intro}
\vspace{-1mm}

Polynomial factorization is a fundamental question in mathematics and computing. In the last decades, quite efficient algorithms have been invented for various fields, e.g., over rationals \cite{lenstra1982factoring}, number fields \cite{landau1985factoring}, finite fields \cite{berlekamp1967factoring, cantor1981new, kedlaya2011fast}, $p$-adic fields \cite{chistov1987efficient, cantor2000factoring}, etc. 
Being a problem of huge theoretical and practical importance, it has been very well studied; for more background refer to surveys, e.g., \cite{kaltofen1992polynomial, von2001factoring, forbes2015complexity}.

The same question over {\em composite} characteristic rings is believed to be computationally hard, e.g.~it is related to integer factoring \cite{shamir1993generation, klivans1997factoring}.
What is less understood is factorization over a local {\em ring}; especially, ones that are the residue class rings of $\Z$ or $\F_q[z]$. A natural variant is as follows.

\smallskip\noindent
\textbf{Problem:} Given a univariate integral polynomial $f(x)$ and a prime power $p^k$, with $p$ prime and $k\in \N$; output a nontrivial factor of $f\bmod p^k$ in randomized poly($\deg f, k\log p$) time.

\smallskip
Note that the polynomial ring $(\mathbb{Z}/\langle p^k\rangle)[x]$ is {\em not} a unique factorization domain. So $f(x)$ may have many, usually {\em exponentially} many, factorizations. 
For example, $x^2+px$ has an irreducible factor $x+\alpha p \bmod p^2$ for each $\alpha \in [p]$ and so $x^2+px$ has exponentially many (wrt $\log p$) irreducible factors modulo $p^2$. 
This leads to a total breakdown in the classical factoring methods. %Hence at first one must be interested in finding at least one such factorization.

\smallskip
{\em We give the first randomized polynomial time algorithm to non-trivially factor (or test for irreducibility) a polynomial $f(x) \bmod p^k$, for $k\leq 4$. 

Additionally, when $f\bmod p$ is power of an irreducible, we provide (\& count) all the lifts $\bmod\ p^k$ ($k\leq 4$) of any factor of $f\bmod p$, in randomized polynomial time.}
\smallskip

Usually, one factors $f(x) \bmod p$ and tries to ``lift'' this factorization to higher powers of $p$. If the former is a coprime factorization then Hensel lifting \cite{Hensel1918} helps us in finding a non-trivial factorization of $f(x) \bmod p^k$ for any $k$. But, when $f(x) \bmod p$ is power of an irreducible then it is not known how to lift to some factorization of $f(x) \bmod p^k$. To illustrate the difficulty let us see some examples (also see \cite{von1996factorization}).

\smallskip\noindent
{\bf Example.} [coprime factor case]
Let $f(x)=x^2+10 x+21$. Then $f\equiv x(x+1) \bmod 3$ and Hensel lemma lifts this factorization uniquely mod $3^2$ as $f(x)\equiv (x+1\cdot3)(x+1+2\cdot3)\equiv (x+3)(x+7) \bmod 9$. This lifting extends to any power of $3$.

\smallskip\noindent
{\bf Example.} [power of an irreducible case] Let $f(x)=x^3+12x^2+3x+36$ and we want to factor it mod $3^3$. Clearly, $f\equiv x^3 \bmod 3$. By brute force one checks that, the factorization $f\equiv x\cdot x^2 \bmod 3$ lifts to factorizations mod $3^2$ as: $x(x^2+3x+3)$, $(x+6)(x^2+6x+3)$, $(x+3)(x^2+3)$. Only the last one lifts to mod $3^3$ as: $(x+3)(x^2+9x+3),(x+12)(x^2+3), (x+21)(x^2+18x+3)$. 

So the big issue is: efficiently determine which factorization out of the exponentially many factorizations mod $p^j$ will lift to mod $p^{j+1}$?

\vspace{-2mm}
\subsection{Previously known results}
\label{sec-known-results}
\vspace{-1mm} 

Using Hensel lemma it is easy to find a non-trivial factor of $f \bmod p^k$ when $f\bmod p$ has two coprime factors. So the hard case is when $f \bmod p$ is power of an irreducible polynomial. The first resolution in this case was achieved by \cite{von1998factoring} assuming that $k$ is ``large''. They assumed $k$ to be larger than the maximum power of $p$ dividing the discriminant of the integral $f$. Under this assumption (i.e.~$k$ is large), they showed that factorization modulo $p^k$ is well behaved and it corresponds to the unique $p$-adic factorization of $f$ (refer $p$-adic factoring \cite{chistov1987efficient, chistov1994algorithm, cantor2000factoring}). To show this, they used an extended version of Hensel lifting (also discussed in \cite{borevich1986number}). Using this observation they could also describe {\em all} the factorizations modulo $p^k$, in a compact data structure.  
The complexity of \cite{von1998factoring} was improved by \cite{Cheng2001}.

The related questions of root finding and root counting of $f\bmod p^k$ are also of classical interest, see \cite{niven2013introduction, apostol2013introduction}. A recent result by \cite[Cor.24]{berthomieu2013polynomial} resolves these problems in randomized polynomial time. Again, it describes {\em all} the roots modulo $p^k$, in a compact data structure.  

Root counting has interesting applications in arithmetic algebraic-geometry, eg.~to compute Igusa zeta function of a univariate integral polynomial
\cite{zuniga2003computing, denef2001newton}. Partial derandomization of root counting algorithm has been obtained by \cite{cheng2017counting, kopp2018randomized} last year; however, a deterministic poly-time algorithm is still unknown.

Going back to factoring $f \bmod p^k$, \cite{von1996factorization} discusses the hurdles when $k$ is small. The factors could be completely unrelated to the corresponding $p$-adic factorization, since an irreducible $p$-adic polynomial could reduce mod $p^k$ when $k$ is small. We give an example from  \cite{von1996factorization}. 

\smallskip\noindent
{\bf Example.}
Polynomial $f(x) = x^2+3^k$ is irreducible over $\Z/\langle 3^{k+1}\rangle$ and so over $3$-adic field. But, it is reducible mod $3^{k}$ as $f\equiv x^2 \bmod 3^k$.

\smallskip
They also discussed that the distinct factorizations are completely different and not nicely related, unlike the case when $k$ is large. An example taken from \cite{von1996factorization} is,

\smallskip\noindent
{\bf Example.}
$f= (x^2+243)(x^2+6)$ is an irreducible factorization over $\Z/\langle 3^6\rangle$. There is another completely unrelated factorization $f=(x+351)(x+135)(x^2+243x+249) \bmod 3^6$.

\smallskip
Many researchers tried to solve special cases, especially when $k$ is constant. The only successful factoring algorithm is by \cite{salagean} over $\Z/\langle p^2\rangle$; it is actually related to Eisenstein criterion for irreducible polynomials. The next case, to factor modulo $p^3$, is unsolved and was recently highlighted in  \cite{sircana2017factorization}.

\vspace{-2mm}
\subsection{Our results }\label{sec-results}
\vspace{-1mm}

We saw that even after the attempts of last two decades we do not have an efficient algorithm for factoring mod $p^3$. Naturally, we would like to first understand the difficulty of the problem when $k$ is constant. In this direction we make  significant progress by devising a unified method which solves the problem when $k=2,3$ or $4$ (and sketch the obstructions we face when $k\ge5$). Our first result is,

\vspace{-1mm}
\begin{theorem}
\label{thm1}
Let $p$ be prime, $k\leq 4$ and $f(x)$ be a univariate integral polynomial. Then, $f(x) \bmod p^k$ can be factored (\& tested for irreducibility) in randomized poly($\deg f, \log p$) time. 
\end{theorem}

\noindent
{\bf Remarks. 1)}
The procedure to factorize $f\bmod p^4$ also factorizes $f\bmod p^3$ and $f\bmod p^2$ (and tests for irreducibility) in randomized poly($\deg f, \log p$) time. This solves the open question of efficiently factoring $f \bmod p^3$ \cite{sircana2017factorization} and gives a more general proof for factoring $f\bmod p^2$ than the one in \cite{salagean}.

{\bf 2)} Our method can as well be used to factor a `univariate' polynomial $f\in \left(\F_p[z]/\langle \psi^k\rangle\right)[x]$, for $k\le4$ and irreducible $\psi(z) \bmod p$, in randomized poly($\deg f, \deg \psi, \log p$) time.

\smallskip
Next, we do more than just factoring $f$ modulo $p^k$ for $k\leq 4$. It is well known that Hensel lemma efficiently gives two (unique) coprime factors of $f(x)$ modulo any prime power $p^k$, given two coprime factors of $f\bmod p$; but it fails to lift when $f$ is power of an irreducible polynomial modulo $p$. We show that our method works in this case to give all the lifts $g(x) \bmod p^k$ (possibly exponentially many) of any given factor $\tilde{g}$ of $f \bmod p$, for $k\leq 4$.

\begin{theorem}
\label{thm2}
Let $p$ be prime, $k\leq 4$ and $f(x)$ be a univariate integral polynomial such that $f\bmod p$ is a power of an irreducible polynomial. Let $\tilde{g}$ be a given factor of $f \bmod p$. Then, in randomized poly($\deg f, \log p$) time, we can compactly describe (\& count) all possible factors of $f(x) \bmod p^k$ which are lifts of $\tilde{g}$ (or report that there is none).
\end{theorem}

\noindent
{\bf Remark.}
Theorem \ref{thm2} can be seen as a significant generalization of Hensel lifting method (Lemma \ref{lemma-hensel}) to $\Z/\langle p^k\rangle$, $k\le4$. To lift a factor $f_1$ of $f \bmod p$, Hensel lemma relies on a cofactor $f_2$ which is coprime to $f_1$. Our method needs no such assumption and it directly lifts a factor $\tilde{g}$ of $f\bmod p$ to (possibly exponentially many) factors $g(x)\bmod p^k$.

\vspace{-2mm}
\subsection{Proof technique-- Root finding over local rings} \label{sec-tech}
\vspace{-1mm}

Our proof involves two main techniques which may be of general interest.

\vspace{1mm}
\noindent
\textbf{Technique 1:} Known factoring methods $\bmod\ p$ work by first reducing the problem to that of root finding $\bmod\ p$. In this work, we efficiently reduce the problem of factoring $f(x)$ modulo the principal ideal $\langle p^k \rangle$ to that of finding roots of some polynomial $E(y) \in (\Z[x])[y]$ modulo a {\em bi-generated} ideal $\langle p^k, \varphi(x)^\ell \rangle$, where $\varphi(x)$ is an irreducible factor of $f(x) \bmod p$. This technique works for all $k\ge1$.

\medskip
\noindent
\textbf{Technique 2:} Next, we find a root of the equation $E(y)\equiv 0 \bmod \langle p^k, \varphi(x)^\ell\rangle$, assuming $k\leq 4$. With the help of the special structure of $E(y)$ we will efficiently find all the roots $y$ (possibly exponentially many) in the local ring $\Z[x]/\langle p^k, \varphi(x)^\ell\rangle$. 

It remains open whether this technique extends to $k=5$ and beyond (even to find a single root of the equation). The possibility of future extensions of our technique is discussed in Appendix~\ref{sec-comparison}.

\vspace{-2mm}
\subsection{Proof overview}
\vspace{-1mm}

\medskip\noindent
\textbf{Proof idea of Theorem \ref{thm1}:} 
Firstly, assume that the given degree $d$ integral polynomial $f$ satisfies $f(x)\equiv \varphi^e \bmod p$ for some $\varphi(x) \in \Z[x]$ which is irreducible mod $p$. Otherwise, using Hensel lemma (Lemma \ref{lemma-hensel}) we can efficiently factor $f\bmod p^k$.

Any factor of such an $f \bmod p^k$ must be of the form $(\varphi^a-p y)\bmod p^k$, for some $1\le a< e$ and $y \in (\Z/\langle p^k\rangle)[x]$. In Theorem \ref{thm-reduction}, we first reduce the problem of finding such a factor $(\varphi^a-p y)$ of $f \bmod p^k$ to finding roots of some $E(y) \in (\Z[x])[y]$ in the local ring $\Z[x]/ \langle p^k,\varphi^{ak}\rangle$. This is inspired by the $p$-adic power series expansion of the quotient $f/(\varphi^a-p y)$. On going $\bmod\ p^k$ we get a polynomial in $y$ of degree $(k-1)$; which we want to be divisible by $\varphi^{ak}$.

The root $y$ of $E(y) \bmod \langle p^k, \varphi^{ak}\rangle$ can be further decomposed into coordinates $y_0, y_1,\ldots$, $y_{k-1} \in \F_p[x]/\langle \varphi^{ak}\rangle$ such that $y=: y_0+p y_1+\ldots+p^{k-1} y_{k-1} \bmod \langle p^k,\varphi^{ak}\rangle$. When we take $k=4$, it turns out that the root $y$ only depends on the coordinates $y_0$ and $y_1$ (i.e.~$y_2, y_3$ can be picked arbitrarily).

Next, we reduce the problem of root finding of $E(y_0+p y_1)$ in the ring $\Z[x]/\langle p^4, \varphi^{4 a}\rangle$ to root finding in characteristic $p$; of some $E'(y_0,y_1)$ in the ring $\F_p[x]/\langle \varphi^{4 a}\rangle$ (Lemma \ref{lemma-A}). We take help of a subroutine \Call{Root-Find}{} given by \cite{berthomieu2013polynomial} which can efficiently find all the roots of a univariate $g(y)$ in the ring $\Z/\langle p^j\rangle$. We need a slightly generalized version of it, to find all the roots of a given $g(y)$ in the ring $\F_p[x]/\langle \varphi(x)^j \rangle$ (Appendix \ref{appendix-rootfind}). 

Note that $y_0, y_1$ are in the ring $\F_p[x]/\langle \varphi^{4a}\rangle$ and so they can be decomposed as $y_0=: y_{0,0}+\varphi y_{0,1}+\ldots+\varphi^{4a-1} y_{0,4a-1}$ and $y_1=: y_{1,0}+\varphi y_{1,1}+\ldots+\varphi^{4a-1} y_{1,4a-1}$, with all $y_{i,j}$'s in the field $\F_p[x]/\langle \varphi\rangle$. 

To get $E'(y_0,y_1) \bmod \langle p,\varphi^{4a}\rangle$ the idea is: to first divide by $p^2$, and then to go modulo the ideal $\langle p,\varphi^{4a}\rangle$. Apply Algorithm \Call{Root-Find}{} to solve $E(y_0+p y_1)/p^2\equiv 0 \bmod \langle p,\varphi^{4a}\rangle$. This allows us to fix some part of $y_0$, say $a_{0} \in \F_p/\langle \varphi^{4a}\rangle$, and we can replace it by $a_{0}+\varphi^{i_0}y_0$, $i_0\ge1$. Thus, $p^3 | E(a_{0}+\varphi^{i_0}y_0+p y_1) \bmod \langle p^4,\varphi^{4a}\rangle$ and we divide out by this $p^3$ (\& change the modulus to $\langle p,\varphi^{4a}\rangle$). In Lemma \ref{lemma-A} we show that when we go modulo the ideal $\langle p,\varphi^{4a}\rangle$ (to find $a_0$), we only need to solve a univariate in $y_0$ using \Call{Root-Find}{}. So, we only need to fix some part of $y_0$, that we called $a_0$, and $y_1$ is irrelevant. Finally, we get $E'(y_0,y_1)$ such that $E'(y_0,y_1):= E(a_0+\varphi^{i_0}y_0+p y_1)/p^3 \bmod \langle p,\varphi^{4a}\rangle$. Importantly, the process yields at most {\em two} possibilities of $E'$ (resp.~$a_0$) to deal with.

Lemma \ref{lemma-A} also shows that the bivariate $E'(y_0,y_1)$ is a special one of the form $E'(y_0,y_1)\equiv E_1(y_0) + E_2(y_0)y_1 \bmod \langle p,\varphi^{4a}\rangle$, where $E_1(y_0) \in (\F_p[x]/\langle \varphi^{4a}\rangle)[y_0]$ is a cubic univariate polynomial and $E_2(y_0) \in (\F_p[x]/\langle \varphi^{4a}\rangle)[y_0]$ is a linear univariate polynomial. We exploit this special structure to represent $y_1$ as a rational function of $y_0$, i.e.~$y_1\equiv -E_1(y_0)/E_2(y_0) \bmod \langle p,\varphi^{4a}\rangle$. The important issue is that, we can calculate $y_1$ only when on some specialization $y_0=a_0$, the division by $E_2(a_0)$ is well defined. So what we do is, we guess each value of $0\le r\le 4a$ and ensure that the valuation (wrt $\varphi$ powers) of $E_1(y_0)$ is at least $r$ but that of $E_2(y_0)$ is {\em exactly} $r$. Once we find such a $y_0$, we can efficiently compute $y_1$ as $y_1\equiv -(E_1(y_0)/\varphi^r)/(E_2(y_0)/\varphi^r) \bmod \langle p,\varphi^{4a-r}\rangle$.

To find $y_0$, we find common solution of the two equations: $E_1(y_0)\equiv E_2(y_0) \equiv 0 \bmod \langle p,\varphi^{r}\rangle$, for each guessed value $r$, using Algorithm \Call{Root-Find}{}. Since the polynomial $E_2(y_0)$ is linear, it is easy for us to filter all $y_0$'s for which valuation of $E_2(y_0)$ is {\em exactly} $r$ (Lemma \ref{lemma-B}). Thus, we could efficiently find all $(y_0,y_1)$ pairs that satisfy the equation $E'(y_0,y_1) \equiv 0 \bmod \langle p,\varphi^{4a}\rangle$.

\medskip\noindent
\textbf{Proof idea of Theorem \ref{thm2}:} 
If $f\equiv \varphi^e \bmod p$ then any lift $g(x)$ of a factor $\tilde{g}(x)\equiv \varphi^a \bmod p$ of $f\bmod p$ will be of the form $g\equiv (\varphi^a-p y) \bmod p^k$. So basically we want to find all the $y$'s $\bmod\ p^{k-1}$ that appear in the proof idea of Theorem \ref{thm1} above. This can be done easily, because Algorithm \Call{Root-Find}{} (Appendix \ref{appendix-rootfind}) \cite{berthomieu2013polynomial} describes all possible $y_0$'s in a compact data structure. Moreover, using this, a count of all $y$'s could be provided as well.

%% file: prelim.tex
%%%%%%%%%%%%%%%%%%%%%%%%%%%%%%%%%%%%%%%%%%%%%%%%%%%%%%%%%%%%%%%%%%%%%%%
\section{Preliminaries}
\label{sec-pre}
\vspace{-1mm}
%%%%%%%%%%%%%%%%%%%%%%%%%%%%%%%%%%%%%%%%%%%%%%%%%%%%%%%%%%%%%%%%%%%%%%%

%\noindent
%\textbf{Notations and definitions:}

%First we will define the binary operations of a scalar with a set of scalars.

Let $R(+,.)$ be a ring and $S$ be a non-empty subset of $R$. The product of the set $S$ with a scalar $a\in R$ is defined as $aS:=\{as\mid s\in S \}$. 
Similarly, the sum of a scalar $u\in R$ with the set $S$ is defined as $u+S:=\{ u+s \mid s\in S \}$. 
Note that the product and the sum operations used inside the set are borrowed from the underlying ring $R$. Also note that if $S$ is the empty set then so are $aS$ and $u+S$ for any $a,u\in R$.

%These product and sum operations on the empty set $\{\}$ have no effect and produce empty set $\{\}$, i.e, $u+\{\}:=\{\}$ and $a\{\}:=\{\}$.

\smallskip {\bf Representatives.}
The symbol `*' in a ring $R$, wherever appears, denotes all of ring $R$. For example, suppose $R= \Z/\langle p^k \rangle$ for a prime $p$ and a positive integer $k$. 
In this ring, we will use the notation $y=y_0+p y_1 + \ldots + p^i y_i + p^{i+1} *$, where $i+1 < k$ and each $y_j \in R/\langle p\rangle$, to denote a set $S_y\subseteq R$ such that
$$ S_y=\{ y_0+ \ldots + p^i y_i + p^{i+1} y_{i+1}+ \ldots + p^{k-1} y_{k-1} \mid \forall y_{i+1},\ldots,y_{k-1} \in R/\langle p\rangle \}. $$
Notice that the number of elements in $R$ represented by $y$ is $|S_y|=p^{k-i-1}$. 

%To write a `series form' as above in a more concise way we define the following notation.

We will sometimes write the set $y=y_0+p y_1 + \ldots + p^i y_i + p^{i+1} *$ succinctly as $y = v + p^{i+1} *$, where $v\in R$ stands for $v=y_0+p y_1 + \ldots + p^i y_i$. 
%This is to avoid long series notation (if it is not worth emphasizing). So readers should not be confused at looking this notation.

%We would like to point out that our definition of $*$ notation is quite flexible and it should not be treated exactly like a variable (except for the operation defined on this). It is merely a compact form of writing ``every thing from ring works at this place" and the meaning of $*$ depends on the underlying modular ring and the precision of the value using it. So $*$ in $u= u_0 + p^i * \in \Z/\langle p^k\rangle$ is used to represent $p^{k-i-1}$ values and in $v= v_0 + p^j * \in \Z/\langle p^k\rangle$ it is used to represent $p^{k-j-1}$ values.
%
%With these notations in hand we now will define {\em representative roots} or {\em representative pairs}. An element $a_0+\varphi^{i_0} * \in R_0$ is called a representative root (or corresponding pair $(a_0,i_0)$ is called representative pair) of some polynomial $g(y)\in R_0$ if $g(a_0+\varphi^{i_0} y)\equiv 0 \bmod \langle p,\varphi^{i} \rangle$ for all $i\geq i_0$ and $g(\tilde{a}_0+\varphi^{i_0-1} y)\not \equiv 0 \bmod \langle p,\varphi^{i_0}\rangle$ where $\tilde{a_0} \equiv a_0 \bmod \langle p,\varphi^{i_0-1}\rangle$.

In the following sections, we will add and multiply the set $\{*\}$ with scalars from the ring $R$. Let us define these operations as follows ($*$ is treated as an unknown)
\begin{itemize}
	\item $u+\{*\}:=\{ u+* \}$ and $u\{ *\}:=\{ u*\}$, where $u\in R$.
	\item $c+\{a+{b*}\}=\{ (a+c)+{b*}\}$ and $c\{a+{b*}\}=\{ac+{bc*} \}$, where $a,b,c \in R$.
\end{itemize}

Another important example of the $*$ notation: Let $R= \F_p[x]/\langle\varphi(x)^k\rangle$ for a prime $p$ and an irreducible $\varphi\bmod p$. 
In this ring, we use the notation $y=y_0+ \varphi y_1 + \ldots + \varphi^i y_i + \varphi^{i+1} *$, where $i+1 < k$ and each $y_j \in R/\langle \varphi\rangle$, to denote a set $S_y\subseteq R$ such that
\[ S_y=\{ y_0+ \ldots + \varphi^i y_i + \varphi^{i+1} y_{i+1}+ \ldots + \varphi^{k-1} y_{k-1} \mid \forall y_{i+1},\ldots,y_{k-1} \in R/\langle \varphi\rangle \}.\]

\smallskip {\bf Zerodivisors.}
Let $R[x]$ be the ring of polynomials over $R=\Z/\langle p^k \rangle$. The following lemma about zero divisors in $R[x]$ will be helpful.

\begin{lemma}
\label{lemma-unique-factor}
A polynomial $f \in R[x]$ is a zero divisor iff $f \equiv 0 \bmod p$. Consequently, for any polynomials $f,g_1,g_2 \in R[x]$ and $f \not\equiv 0 \bmod p$, $f(x) g_1(x) = f(x) g_2(x)$ implies $g_1(x) = g_2(x)$. 
\end{lemma}
\begin{proof}
If $f \equiv 0 \bmod p$ then $f(x) p^{k-1}$ is zero, and $f$ is a zero divisor. 

For the other direction, let $f \not\equiv 0 \bmod p$ and assume $f(x)g(x) = 0$ for some non-zero $g\in R[x]$. Let
\begin{itemize}
	\item $i$ be the biggest integer such that the coefficient of $x^i$ in $f$ is non-zero modulo $p$,
	\item and $j$ be the biggest integer such that the coefficient of $x^j$ in $g$ has minimum valuation with respect to $p$.
\end{itemize}
Then, the coefficient of $x^{i+j}$ in $f\cdot g$ has same valuation as the coefficient of $x^j$ in $g$, implying that the coefficient is nonzero. This contradicts the assumption $f(x)g(x) = 0$. 

The consequence follows because $f \not \equiv 0 \bmod p$ implies that $f$ cannot be a zero divisor.
\end{proof}

%%%%%%%%%%%%%%%%%%%%%%%%%%%%%%%%%%%%%%%%%%%%%%%%%%%%%%%%
% Move this to preliminaries
\smallskip {\bf Quotient ideals.}
We define the quotient ideal (analogous to division of integers) and look at some of its properties. 

\begin{definition}[Quotient Ideal]
Given two ideals $I$ and $J$ of a commutative ring $R$, we define the quotient of $I$ by $J$ as,
\[ I:J := \{ a\in R \mid aJ \subseteq I \}. \]
\end{definition}

It can be easily verified that $I:J$ is an ideal. Moreover, we can make the following observations about quotient ideals.

\begin{claim}[Cancellation]
\label{claim-1}
Suppose $I$ is an ideal of ring $R$ and $a,b,c$ are three elements in $R$. By definition of quotient ideals, $ca\equiv cb \bmod I$ iff $a\equiv b \bmod I:\langle c\rangle$.
\end{claim}
%\begin{proof}
%We know that $p a\equiv p b \bmod I$ iff $p a - p b \in I$. By the definition of quotient ideals $I:\langle p \rangle = \{c\in R \mid cp\in I \}$ and since $(a-b)p \in I$ so clearly $a-b \in I:\langle p\rangle$. And hence $a\equiv b \bmod I:\langle p\rangle$.
%\end{proof}
%
\begin{claim}
\label{claim-2}
Let $p$ be a prime and $\varphi \in (\Z/\langle p^k\rangle)[x]$ be such that $\varphi \not\equiv 0 \bmod p$. Given an ideal $I:=\langle p^l,\phi^m \rangle$ of $\Z[x]$,
\begin{enumerate}
\item $I:\langle p^i\rangle = \langle p^{l-i},\phi^m\rangle$, for $i\le l$, and 
\item $I:\langle \phi^j\rangle = \langle p^l,\phi^{m-j}\rangle$, for $j\le m$.
\end{enumerate}

\end{claim}
\begin{proof}
We will only prove part $(1)$, as proof of part $(2)$ is similar. 
If $c \in \langle p^{l-i}, \varphi^{m}\rangle$ then there exists $c_1,c_2 \in \Z[x]$, such that, $c= c_1 p^{l-i} + c_2 \varphi^m$. Multiplying by $p^i$,
\[ p^i c= c_1 p^{l} + c_2 p^i \varphi^m \in I \Rightarrow c \in I:\langle p^i \rangle.\]
To prove the reverse direction, if $c \in I: \langle p^i \rangle$ then there exists $c_1,c_2 \in \Z[x]$, such that, $p^i c= c_1 p^{l} + c_2 \varphi^m$. 
Since $i\le l$ and $p\not |\ \varphi$, we know $p^i | c_2$. So,

$c= c_1 p^{l-i} + (c_2/ p^i) \varphi^m \ \Rightarrow\  c\in \langle p^{l-i},\phi^m\rangle.$   	
\end{proof}

\begin{lemma}[Compute quotient]
\label{lemma-compute-g-modp}
Given a polynomial $\varphi \in \Z[x]$ not divisible by $p$, define $I$ to be the ideal $\langle p^l,\phi^m\rangle$ of $\Z[x]$. 
If $g(y) \in (\Z[x])[y]$ is a polynomial such that $g(y) \equiv 0 \bmod \langle p,\phi^m \rangle$, then $p | g(y) \bmod I$ and $g(y)/p \bmod I:\langle p\rangle$ is efficiently computable.
\end{lemma}

\begin{proof}
The equation $g(y)\equiv 0 \bmod \langle p, \phi^m\rangle$ implies $g(y) = p c_1(y)  + \varphi^m c_2(y)$ for some polynomials $c_1(y), c_2(y) \in \Z[x][y]$. 
Going modulo $I$, $g(y)\equiv p c_1(y) \bmod I $. Hence, $p | g(y) \bmod I$ and $g(y)/p\equiv c_1(y) \bmod I:\langle p\rangle$ (Claim \ref{claim-1}).

If we write $g$ in the reduced form modulo $I$, then the polynomial $g(y)/p$ can be obtained by dividing each coefficient of $g(y)\bmod I$ by $p$. 
%(since $\varphi$ is monic and irreducible mod $p$ so $p\not | \varphi \bmod I$).
\end{proof}

% upto this part should be moved to preliminaries
%%%%%%%%%%%%%%%%%%%%%%%%%%%%%%%%%%%%%%%%%%%%%%%%%%%%%%%%%%%%%%%%%%%%%%%%%%%%%%%%%%%%%%%%%% 

%The following lemma shows that any factor $h(x)\equiv (\varphi^a-p y) \in (\Z/\langle p^k\rangle)[x]$ of $f$ will have a unique co-factor $g(x) \in (\Z/\langle p^k\rangle)[x]$ of $f(x)\bmod p^k$.
%\begin{lemma}
%\label{lemma-unique-factor}
%Let $f \in \Z[x]$ is as assumed above. If $\exists h(x) \in (\Z/\langle p^k\rangle)[x]$ such that $h(x) \equiv (\varphi^a-p y) \bmod p^k$ where $a\leq e$ and $h(x)| f(x) \bmod p^k$ then there exists a unique polynomial $g(x) \in (\Z/\langle p^k\rangle)[x]$ such that $g(x)\not\equiv 0 \bmod p$ and $f(x) \equiv h(x) g(x) \bmod p^k$ .
%\end{lemma}
%\begin{proof}
%Let there exists two polynomials $g_1(x), g_2(x) \in \Z/\langle p^k\rangle$ with $g_1(x)\not\equiv 0 \bmod p$ and $g_2(x)\not\equiv 0 \bmod p$ such that $f(x) \equiv h(x) g_1(x) \bmod p^k$ and $f(x) \equiv h(x) g_2(x) \bmod p^k$. Then $h(x) (g_1(x)-g_2(x))\equiv 0 \bmod p^k$. Since $h(x)\not \equiv 0 \bmod p$ and so it can not be a zero-divisor in ring $(\Z/\langle p^k\rangle)[x]$ and so $g_1(x) \equiv g_2(x) \bmod p^k$.
%\end{proof}

%% file: sec_3-1.tex
%%%%%%%%%%%%%%%%%%%%%%%%%%%%%%%%%%%%%%%%%%%%%%%%%%%%%%%%%%%%%%%%%%%%%%%
\section{Main Results: Proof of Theorems \ref{thm1} and \ref{thm2}}
 \label{sec-main}
\vspace{-1mm}
%%%%%%%%%%%%%%%%%%%%%%%%%%%%%%%%%%%%%%%%%%%%%%%%%%%%%%%%%%%%%%%%%%%%%%%

%%%%%%%%%%%%%%%%%%%%%%%%%%%%%%%%%%%%%%%%%%%%%%%%%%%%%%%%%%%%%%%%%%%%%%%
%\subsection{Preprocessing on the input}
%\label{sec-notation}
%\vspace{-1mm}

%We will state the assumptions which we make for our input and do some initial preprocessing on the input.

Our task is to factorize a univariate integral polynomial $f(x)\in \Z[x]$ of degree $d$ modulo a prime power $p^k$. Without loss of generality, we can assume that $f(x) \not \equiv  0 \bmod p$. 
Otherwise, we can efficiently divide $f(x)$ by the highest power of $p$ possible, say $p^l$, such that $f(x) \equiv p^l \tilde{f}(x) \bmod p^k$ and $\tilde{f}(x) \not \equiv 0 \bmod p$.
In this case, it is equivalent to factorize $\tilde{f}$ instead of $f$.

To simplify the input further, write $f \bmod p$ (uniquely) as a product of powers of coprime irreducible polynomials. 
If there are two coprime factors of $f$, using Hensel lemma (Lemma \ref{lemma-hensel}), we get a non-trivial factorization of $f$ modulo $p^k$. So, we can assume that $f$ is a power of a monic irreducible polynomial $\varphi\in\Z[x]$ modulo $p$. In other words, we can efficiently write $f\equiv \varphi^e + p l \bmod p^k$ for a polynomial $l$ in $(\Z/\langle p^k\rangle) [x]$. We have $e\cdot\deg\varphi\le \deg f$, for the integral polynomials $f$ and $\varphi$.

\subsection{Factoring to Root-finding}
\label{fact-root}
\vspace{-1mm}

By the preprocessing above, we only need to find factors of a polynomial $f$ such that $f\equiv \varphi^e + p l \bmod p^k$, where $\varphi$ is an irreducible polynomial modulo $p$. Up to multiplication by units, any nontrivial factor $h$ of $f$ has the form $h \equiv \varphi^a - p y$, where $a < e$ and $y$ is a polynomial in $(\Z/\langle p^k\rangle) [x]$.
   
Let us denote the ring $\Z[x]/\langle p^k,\varphi^{ak} \rangle$ by $R$. Also, denote the ring $\Z[x]/\langle p,\varphi^{ak} \rangle$ by $R_0$. 
We define an auxiliary polynomial $E(y) \in R[y]$ as 
$$ E(y) := f(x) (\varphi^{a(k-1)} + \varphi^{a(k-2)}(py)+ \ldots + \varphi^a (py)^{k-2}+ (py)^{k-1}). $$ 

%The definitions of $E(y), R$ and $R_0$ depends on the value of $k$ and so it will matter where we are using %these terms which should be clear by the context. 

Our first step is to reduce the problem of factoring $f(x) \bmod p^k$ to the problem of finding roots of the  univariate polynomial $E(y)$ in $R$. 
Thus, we convert the problem of finding factors of $f(x)\in \Z[x]$ modulo a principal ideal $\langle p^k\rangle$ to root finding of a polynomial $E(y) \in (\Z[x])[y]$ modulo a bi-generated ideal $\langle p^k, \varphi^{ak} \rangle$.  % This line could be removed.

\begin{theorem}[Reduction theorem]
\label{thm-reduction}
Given a prime power $p^k$; let $f(x),h(x)\in \mathbb{Z}[x]$ be two polynomials of the form $f(x)\equiv \varphi^e+ p l \bmod p^k$ and $h(x)\equiv \varphi^a-p y \bmod p^k$.
Here $y,l$ are elements of $(\Z/\langle p^k\rangle) [x]$ and $a\leq e$. Then, $h$ divides $f$ modulo $p^k$ if and only if
$$ E(y) \ =\ f(x) (\varphi^{a(k-1)} + \varphi^{a(k-2)}(py)+ \ldots + \varphi^a (py)^{k-2}+ (py)^{k-1}) \ \equiv\ 0 \bmod \langle p^k,\varphi^{ak}\rangle . $$
\end{theorem}

\begin{proof}
Let $Q$ denote the {\em ring of fractions} of the ring $(\Z/\langle p^k\rangle)[x]$. Since $\varphi$ is not a zero divisor, $(E(y)/\varphi^{ak}) \in Q$.

We first prove the reverse direction. If $E(y)\equiv 0 \bmod \langle p^k,\varphi^{ak}\rangle$, then $(E(y)/\varphi^{ak})$ is a valid polynomial over $(\Z/\langle p^k\rangle)[x]$. 
Multiplying $h$ with $(E(y)/\varphi^{ak}) \bmod p^k$, we write,
\[(\varphi^a-py) ( (f/\varphi^{ak}) \Sigma_{i=0}^{k-1}\varphi^{a(k-1-i)}(py)^i )\ \equiv\ (f/\varphi^{ak}) ( \varphi^{ak}-(py)^k ) \equiv f\cdot \varphi^{ak}/\varphi^{ak} \equiv f \bmod p^k.\]
Hence, $h$ divides $f$ modulo $p^k$.

%So we have proved that if $E(y)\equiv 0 \bmod \langle p^k,\varphi^{ak}\rangle$ then there is a valid polynomial $g(x)$ in $(\Z/\langle p^k\rangle)[x]$ such that $f(x)$ factors as $f(x)\equiv h(x)g(x) \bmod p^k$. Now we will prove the opposite direction $`\Rightarrow '$.
\smallskip
For the forward direction, assume that there exists some $g(x) \in (\Z/\langle p^k\rangle)[x]$, such that, $f(x)\equiv h(x) g(x) \bmod p^k$. We get two factorizations of $f$ in $Q$,
\[ f(x) = h(x) g(x) \text{~~~  and  ~~~} f(x) = h(x) (E(y)/\varphi^{ak}) .\]

Subtracting the first equation from the second one,
\[ h(x) \left( g(x) - (E(y)/\varphi^{ak}) \right) = 0 .\]

Notice that $h(x)$ is not a zero divisor in $(\Z/\langle p^k \rangle) [x]$ (by Lemma
\ref{lemma-unique-factor}) and is thus invertible in $Q$. 
	So, $E(y)/\varphi^{ak} =g(x)$ in $Q$. Since $g(x)$ is in $(\Z/\langle p^k \rangle) [x]$, we deduce the equivalent divisibility statement: $E(y) \equiv 0 \bmod \langle p^k,\varphi^{ak} \rangle$.
\end{proof}	
%It is enough to show that such a $g(x)$ is unique by showing that $g(x)\in Q$ is unique, since replacing $g(x) \in Q$ by $(E(y)/\varphi^{ak}) \in Q$ satisfies the equation $f(x)\equiv (\varphi^a-py) g(x)$ over $Q$ and hence if this is satisfied by $g(x) \in (\Z/\langle p^k\rangle)[x]$ then $g(x) \equiv (E(y)/\varphi^{ak}) \in (\Z/\langle p^k\rangle)[x]$ and so $E(y)\equiv 0 \bmod \langle p^k,\varphi^{ak}\rangle$.
%
%Assume that there exists some $g_1(x) \in Q$ such that $f(x)\equiv h(x) g_1(x) \in Q$. Then $h(x) (g(x) -g_1(x)) \equiv 0 \in Q$. Since $h(x)\equiv (\varphi^a - p y) \in \Z/\langle p^k\rangle[x]$ is not a zero divisor and so not a zero divisor in $Q$ hence $g(x)\equiv g_1(x) \in Q$. So there is unique $g(x) \in \Z/\langle p^k\rangle[x]$.

%So now our focus shifts from finding the factors of $f(x)\bmod p^k$ to find roots of $E(y) \bmod \langle p^k,\varphi^{ak} \rangle$ where $E(y) \in R[y]$ is as defined in section \ref{sec-notation}. 

The following two observations simplify our task of finding roots $y$ of polynomial $E(y)$.
\begin{itemize}
	\item  First, due to symmetry, it is enough to find factors $h \equiv \varphi^a \bmod p$ with $a\leq e/2$. 
	       The assertion follows because $f\equiv h g \bmod p^k$ implies, at least one of the factor (say $h$) must be of the form $\varphi^a \bmod p$ for $a\leq e/2$. 
	       By Lemma \ref{lemma-unique-factor}, for a fixed $h \equiv (\varphi^a -p y) \bmod p^k$, there is a unique $g\equiv (\varphi^{e-a}-p y') \bmod p^k$ such that $f\equiv h g \bmod p^k$. 
	       So, to find $g$, it is enough to find $h$.
       \item  Second, observe that any root $y \in R$ (of $E(y)\in R[y]$) can be seen as $y= y_0 + p y_1 + p^2 y_2 + \ldots + p^{k-1} y_{k-1}$, where each $y_i \in R_0$ for all $i$ in $\{0,\ldots,k-1\}$. 
%In the next section, Sec.~\ref{sec-fac-p4}, we will find $y$ term by term. 
%In other words, we will find $y_0,y_1,y_2, \cdots$ modulo $\langle p,\varphi^{ak} \rangle$ one by one, to construct the complete root $y$ modulo $\langle p^k,\varphi^{ak} \rangle$. 
	       The following lemma decreases the required precision of root $y$.
\end{itemize}
\begin{lemma}
\label{lemma-1}   % we can change its name to lemma-red-precision, and change it to a claim instead of lemma ??
Let $y= y_0 + p y_1 + p^2 y_2 + \ldots + p^{k-1} y_{k-1}$ be a root of $E(y)$, where $k \geq 2$ and $a\leq e/2$.  
Then, all elements of set $y= y_0 + p y_1 + p^2 y_2 + \ldots + p^{k-3} y_{k-3} + p^{k-2} *$  are also roots of $E(y)$.
\end{lemma}
\begin{proof}
Notice that the variable $y$ is multiplied with $p$ in $E(y)$, implying $y_{k-1}$ is irrelevant. 
Similar argument is applicable for the variable $y_{k-2}$ in any term of the form $(py)^i$ for $i\geq 2$. The only remaining term containing $y_{k-2}$ is $f \varphi^{a(k-2)} (py)$. 
The coefficient of $y_{k-2}$ in this term is $\varphi^{a(k-2)} f p^{k-1}$.
This coefficient vanishes modulo $\langle p^k,\varphi^{ak}\rangle$ too, because 
 
$\varphi^{a(k-2)} f \,\equiv\, \varphi^{a(k-2)} \varphi^e \,\equiv\, \varphi^{ak}\varphi^{e-2a} \,\equiv\, 0 \bmod \langle p,\varphi^{ak}\rangle \ .$ 
\end{proof}

{\bf Root-finding modulo a principal ideal.}
Finally, we state a slightly modified version of the theorem from \cite[Cor.24]{berthomieu2013polynomial}, showing that all the roots of a polynomial $g(y)\in R_0[y]$ can be efficiently described. They gave their algorithm to find (all) roots in $\Z/\langle p^n\rangle$; we modify it in a straightforward way to find (all) roots in $\F_p[x]/\langle\varphi^{ak}\rangle = R_0$ (Appendix \ref{appendix-rootfind}).
Any root in $R_0$ can be written as $y = y_0+\varphi y_1 + \cdots + \varphi^{ak-1} y_{ak-1}$, where each $y_j$ is in the field $R_0/\langle \varphi \rangle$.

Let $g(y)$ be a polynomial in $R[y]$, then a set $y=y_0+\varphi y_1 + \ldots + \varphi^i y_i + \varphi^{i+1} *$ will be called a \emph{representative root} of $g$ iff
\begin{itemize}
	\item All elements in $y=y_0+\varphi y_1 + \ldots + \varphi^i y_i + \varphi^{i+1} *$ are roots of $g$.
	\item Not all elements in $y'=y_0+\varphi y_1 + \ldots + \varphi^{i-1} y_{i-1} + \varphi^i *$ are roots of $g$.
\end{itemize}

We will sometimes represent the set of roots, $y=y_0+\varphi y_1 + \ldots + \varphi^i y_i + \varphi^{i+1} *$, succinctly as $y = v + \varphi^{i+1} *$, 
where $v\in R$ stands for $y=y_0+\varphi y_1 + \ldots + \varphi^i y_i$. Such a pair, $(v,i+1)$, will be called a \emph{representative pair}.

%This is a slightly modified version of the result of Berthomieu et al \cite{berthomieu2013polynomial}. Basically, it says that there can be at most $\deg(g)$ many representative roots (of the form $a_0+\varphi^{i_0} *$) of $g(y)$ in $R_0$ and it will return all such representative pairs $(a_0,i_0)$ in randomized poly($\deg(g), \log p, k$)-time.

\begin{theorem}{\cite[Cor.24]{berthomieu2013polynomial}}
\label{thm-random-root-find}
Given a bivariate $g(y) \in R_0[y]$ where $R_0 = \Z[x]/\langle p, \varphi^{ak}\rangle$, let $Z \subseteq R_0$ be the root set of $g(y)$. Then $Z$ can be expressed as the disjoint union of at most $\deg_y(g)$ many representative pairs $(a_0,i_0)$ ($a_0 \in R_0$ and $i_0 \in \N$).

These representative pairs can be found in randomized poly($\deg_y(g), \log p, ak\deg\varphi$) time.
%such that $g(a_0+\varphi^{i_0} *) \equiv 0 \bmod \langle p, \varphi^{ak}\rangle$.
\end{theorem}

For completeness, Algorithm \Call{Root-Find}{$g,R_0$} is given in Appendix \ref{appendix-rootfind}.

\smallskip
We will fix $k=4$ for the rest of this section. Similar techniques (even simpler) work for $k=3$ and $k=2$. 
The issues with this approach for $k>4$ will be discussed in Appendix~\ref{sec-comparison}. 
%Since the definition of $E(y), R$ and $R_0$ depends on the value of $k$ (as pointed out in section \ref{sec-notation}) they must be seen accordingly.

%% file: sec_3-2.tex
%%%%%%%%%%%%%%%%%%%%%%%%%%%%%%%%%%%%%%%%%%%%%%%%%%%%%%%%
\subsection{Reduction to root-finding modulo a principal ideal of $\F_p[x]$}
\label{sec-fac-p4}
\vspace{-1mm}

In this subsection, the task to find roots of $E(y)$ modulo the bi-generated ideal $\gen{p^4,\varphi^{4a}}$ of $\Z[x]$ will be reduced to finding roots modulo the principal ideal $\gen{\varphi^{4a}}$ (of $\F_p[x]$).

%Now we will explain the idea to reduce the root finding of $E(y) \bmod \langle p^4, \varphi^{4a}\rangle$ to root finding of $E'(y_0,y_1) \bmod \langle p,\varphi^{4a}\rangle$.

Let us consider the equation $E(y)\equiv 0 \bmod \gen{p^4,\varphi^{4a}}$. We have,
\begin{equation}
\label{eq1}
	f (\varphi^{3a}+\varphi^{2a} (py)+\varphi^{a} (py)^2+(py)^3)\equiv 0 \bmod \gen{p^4,\varphi^{4a}} \,.
\end{equation}

Using Lemma \ref{lemma-1}, we can assume $y= y_0 + p y_1$,  
\begin{equation}
\label{eq2}
	f (\varphi^{3a}+\varphi^{2a} p(y_0+p y_1)+\varphi^{a} p^2 (y_0^2 + 2 p y_0 y_1)+(p y_0)^3)\equiv 0 \bmod \gen{p^4,\varphi^{4a}} \,.
\end{equation}

The idea is to first solve this equation modulo $\gen{p^3, \varphi^{4a}}$. Since $f \equiv \varphi^{e} \bmod p$, $e \geq 2a$, variable $y_1$ is redundant while solving this equation modulo $p^3$.
Following lemma finds all representative pairs $(a_0,i_0)$ for $y_0$, such that, $E(a_0+\varphi^{i_0}y_0+py_1) \equiv 0 \bmod \gen{p^3, \varphi^{4a}}$ for all $y_0,y_1 \in R$. Alternatively, we can state this in the polynomial ring $R[y_0,y_1]$.
Dividing by $p^3$, we will be left with an equation modulo the principal ideal $\gen{\varphi^{4a}}$ (of $\F_p[x]$).

%In this stage, we will find the values of $y$ partially as $y= (a_0+\varphi^i y_0) + p y_1$ for some $a_0 \in R_0$ and $i\in \N$ so that $E(y)$ is divisible by $p^3$. Note that we only need to find $y_0$ partially, $y_1$ does not matter. This will give us $E'(y_0,y_1)=E_1(y_0)+E_2(y_0) y_1 \bmod \langle p, \phi^{4a}\rangle$ where $E_1 $ is cubic in $y_0$ and $E_2$ is linear in $y_0$. The following lemma says that we can find all such $y_0=a_0+\varphi^i * \subseteq R_0$ efficiently.
%

\begin{lemma}[Reduce to char=$p$]
\label{lemma-A}
We efficiently compute a unique set $S_0$ of all representative pairs $(a_0,i_0)$, where $a_0\in R_0$ and $i_0\in \N$, such that, 
\[ E((a_0+\varphi^{i_0} y_0) + p y_1) \ =\ p^3 E'(y_0,y_1) \ \bmod \gen{p^4,\varphi^{4a}} \] 
for a polynomial $E'(y_0,y_1) \in R_0[y_0,y_1]$ (it depends on $(a_0,i_0)$). Moreover, 
\begin{enumerate}
	\item $|S_0|\leq 2$. If our efficient algorithm fails to find $E'$ then Eqn.~\ref{eq2} has no solution.
	\item $E'(y_0,y_1) =: E_1(y_0)+E_2(y_0)y_1$, where $E_1(y_0) \in R_0[y_0]$ is cubic in $y_0$ and $E_2(y_0) \in R_0[y_0]$ is linear in $y_0$.
	\item For every root $y\in R$ of $E(y)$ there exists $(a_0, i_0) \in S_0$ and $(a_1,a_2) \in R \times R$, such that $y=(a_0+\varphi^{i_0} a_1) + p a_2$ and		
	$E'(a_1,a_2)\equiv 0 \bmod \gen{p,\varphi^{4a}}$.
\end{enumerate}
\end{lemma}
\noindent We think of $E'$ as the quotient $E((a_0+\varphi^{i_0} y_0) + p y_1)/p^3$ in the polynomial ring $R_0[y_0,y_1]$; and would work with it instead of $E$ in the root-finding algorithm.
\begin{proof}
Looking at Eqn.~\ref{eq2} modulo $p^2$,
\[  f \varphi^{2a} (\varphi^{a} + p y_0) \ \equiv\ 0\ \bmod \gen{p^2,\varphi^{4a}} . \]
Substituting $f = \varphi^{e} + p h_1$, we get $(\varphi^e + ph_1) (\varphi^{3a} + \varphi^{2a} p y_0) \ \equiv\ 0 \bmod \gen{p^2, \varphi^{4a}}$. Implying, $p h_1 \varphi^{3a} \ \equiv\ 0 \bmod \gen{p^2, \varphi^{4a}}$.
Using Claim~\ref{claim-2} the above equation implies that,
\begin{equation}
\label{eq-div-p}
h_1 \,\equiv\, 0 \,\bmod \gen{p,\varphi^{a}} \,,
\end{equation}
is a necessary condition for $y_0$ to exist.

We again look at Eqn.~\ref{eq2}, but modulo $p^3$ now: 
$f (\varphi^{3a} + \varphi^{2a} p y_0 + \varphi^{a} p^2 y_0^2) \,\equiv\, 0\, \bmod \gen{p^3,\varphi^{4a}}$.

Notice that $y_1$ is not present because its coefficient: $p^2 f \varphi^{2a} \ \equiv\ 0 \bmod \gen{p^3, \varphi^{4a}}$. Substituting $f = \varphi^{e} + p h_1$, we get,
\[ (\varphi^e + ph_1) (\varphi^{3a} + \varphi^{2a} p y_0 + \varphi^a p^2 y_o^2) \ \equiv\ 0 \bmod \gen{p^3, \varphi^{4a}} .\]
Removing the coefficients of $y_0$ which vanish modulo $\gen{p^3, \varphi^{4a}}$,
\[ \varphi^{e+a} p^2 y_0^2 + \varphi^{3a} p h_1 + \varphi^{2a} p^2 h_1 y_0   \ \equiv\ 0\ \bmod \gen{p^3, \varphi^{4a}}	.\]
From Eqn.~\ref{eq-div-p}, $h_1$ can be written as $p h_{1,1}+\varphi^a h_{1,2}$, so
\[  p^2 \left( \varphi^{e+a} y_0^2 + \varphi^{3a} h_{1,2} y_0 + \varphi^{3a} h_{1,1} \right) \,\equiv\, 0\, \bmod \gen{p^3, \varphi^{4a}} .\]

We can divide by $p^2\varphi^{3a}$ using Claim~\ref{claim-2} to get an equation modulo $\varphi^{a}$ in the ring $\F_p[x]$. 
This is a quadratic equation in $y_0$. Using Theorem \ref{thm-random-root-find}, we find the solution set $S_0$ with at most two representative pairs: for $(a_0,i_0)\in S_0$, every $y\in a_0+\varphi^{i_0}* + p*$ satisfies,
\[ E(y) \ \equiv\ 0\ \bmod \gen{p^3,\varphi^{4a}}\ .  \]

In other words, on substituting $(a_0 + \varphi^{i_0} y_0 + py_1)$ in $E(y)$,
\[ E(a_0+\varphi^{i_0} y_0 + py_1) \ \equiv\ p^3 E'(y_0,y_1) \ \bmod \gen{p^4,\varphi^{4a}},  \]
for a ``bivariate'' polynomial $E'(y_0,y_1) \in R_0[y_0,y_1]$. This sets up the correspondence between the roots of $E$ and $E'$.

Substituting $(a_0 + \varphi^{i_0} y_0 + py_1)$ in Eqn.~\ref{eq2}, we notice that $E'(y_0,y_1)$ has the form $E_1(y_0) + E_2(y_0) y_1$ for a linear $E_2$ and a cubic $E_1$. 

Finally, this reduction is constructive, because of Lemma \ref{lemma-compute-g-modp} and Theorem \ref{thm-random-root-find}, giving a randomized poly-time algorithm. 
\end{proof}

%% file: sec_3-3.tex
\subsection{Finding roots of a \emph{special} bi-variate $E'(y_0,y_1)$ modulo $\gen{p, \varphi^{4a}}$}
\label{sec-root-E'}

The final obstacle is to find roots of $E'(y_0,y_1)$ modulo $\gen{\varphi^{4a}}$ in $\F_p[x]$. 
The polynomial $E'(y_0,y_1)= E_1(y_0) + E_2(y_0) y_1$ is \emph{special} because $E_2 \in R_0[y_0]$ is {\em linear} in $y_0$. 

%Basically we will get all pairs $(y_0,y_1)=(a_1,a_2) \in (R_0)^2$, where $a_1$ and $a_2$ such that $E_1(a_1)+E_2(a_1) a_2 \equiv 0 \bmod \langle p, \varphi^{4a} \rangle$. 

For a polynomial $u\in\F_p[x][\mathbf y]$ we define {\em valuation} $\val_{\varphi}(u)$ to be the largest $r$ such that $\varphi^r | u$.
Our strategy is to go over all possible valuations $0\le r\le 4a$ and find $y_0$, such that,
\begin{itemize}
\item $E_1(y_0)$ has valuation at least $r$. 
\item $E_2(y_0)$ has valuation exactly $r$.
\end{itemize}

From these $y_0$'s, $y_1$ can be obtained by `dividing' $E_1(y_0)$ with $E_2(y_0)$. The lemma below shows that this strategy captures all the solutions. 

%get $y_0=a_1 \in R_0$ such that $y_1 \equiv a_2 \equiv - E_1(a_1)/ E_2(a_1) \bmod \langle p,\varphi^{4a} \rangle$ given that $E_2(a_1)$ is not a zero divisor. This philosophy is captured in the following lemma.

\begin{lemma}[Bivariate solution]
\label{lemma-2}
A pair $(u_0,u_1) \in R_0 \times R_0$ satisfies an equation of the form  $E_1(y_0)+E_2(y_0) y_1 \equiv 0 \bmod \gen{p, \varphi^{4a}}$ if and only if $\val_{\varphi}(E_1(u_0)) \geq \val_{\varphi}(E_2(u_0))$.
\end{lemma}

\begin{proof}
Let $r$ be $\val_{\varphi}(E_2(u_0))$, where $r$ is in the set $\{0,1,\ldots,4a\}$. 
If $\val_{\varphi}(E_1(u_0)) \geq \val_{\varphi}(E_2(u_0))$ then set $u_1 \equiv - (E_1(u_0)/\varphi^r)/ (E_2(u_0)/\varphi^r) \bmod \gen{p,\varphi^{4a-r}}$. 
The pair $(u_0,u_1)$ satisfies the required equation. (Note: If $r=4a$ then we take $u_1=*$.)

Conversely, if $r':=\val_{\varphi}(E_1(u_0)) < \val_{\varphi}(E_2(u_0))\le  4a$ then, 
for every $u_1$,

$\val_{\varphi}(E_1(u_0)+E_2(u_0)u_1) = r' \,\,\Rightarrow\,\, E_1(u_0)+E_2(u_0)u_1 \,\not\equiv\, 0\, \bmod \gen{p,\varphi^{4a}} \,. $
\end{proof}

%The special form of the polynomial $E'(y_0,y_1)$ helps us in getting all good $a_1$s. The idea is that we already fix the valuation, say $r \in \{0,\ldots,4a-1\}$, and then try to find those $y_0=a_1$ such that the valuation of $E_1(a_1)$ is at least $r$ and the valuation of $E_2(a_1)$ is exactly $r$. Then we filter out only those $a_1$s for which $E_2(a_1)\not \equiv 0 \bmod \langle p,\varphi^{r+1} \rangle$, and so ensuring that the valuation of $E_2(a_1)$ is exactly $r$, by using the following lemma.

We can efficiently find all representative pairs for $y_0$, at most three, such that $E_1(y_0)$ has valuation at least $r$ (using Theorem \ref{thm-random-root-find}).  
The next lemma shows that we can efficiently filter all $y_0$'s, from these representative pairs, that give valuation {\em exactly} $r$ for $E_2(y_0)$.

\begin{lemma}[Reduce to a unit $E_2$]
\label{lemma-B}
Given a linear polynomial $E_2(y_0) \in R_0[y_0]$ and an $r \in [4a-1]$, let $(b,i)$ be a representative pair modulo $\gen{p,\varphi^{r}}$, i.e.,
$ E_2(b + \varphi^{i} *) \equiv 0 \bmod \gen{p,\varphi^{r}}$. 
Consider the quotient $E_2'(y_0) := E_2(b + \varphi^{i} y_0)/\varphi^r$. 

If $E_2'(y_0)$ does not vanish identically modulo $\gen{p,\varphi}$, then there exists at most one $\theta \in R_0/\gen{\varphi}$ such that
$E_2'(\theta) \equiv 0 \bmod \gen{p, \varphi}$, and this $\theta$ can be efficiently computed.
\end{lemma}
\begin{proof}
Suppose $E_2(b + \varphi^{i} y_0) \equiv u + v y_0 \equiv 0 \bmod \gen{p, \varphi^{r}}$. Since $y_0$ is formal, we get $\val_{\varphi}(u)\geq r$ and $\val_{\varphi}(v) \geq r$. 
We consider the three cases (wrt these valuations),
\begin{enumerate}
\item $\val_{\varphi}(u)\ge r$ and $\val_{\varphi}(v)=r$: $E_2'(\theta) \not\equiv 0 \bmod \gen{p, \varphi}$, for all $\theta \in R_0/\gen{\varphi}$ except 
	 $\theta = (-u/\varphi^r)/(v/\varphi^r) \bmod \gen{p,\varphi}$.
\item $\val_{\varphi}(u)=r$ and $\val_{\varphi}(v)>r$: $E_2'(\theta) \not\equiv 0 \bmod \gen{p, \varphi}$, for all $\theta \in R_0/\gen{\varphi}$.
\item $\val_{\varphi}(u)>r$ and $\val_{\varphi}(v)>r$: $E_2'(y_0)$ vanishes identically modulo $\gen{p,\varphi}$, so this case is ruled out by the hypothesis.
\end{enumerate}
There is an efficient algorithm to find $\theta$, if it exists; because the above proof only requires calculating valuations which entails division operations in the ring.
\end{proof}

%Again, similar to Lem.~\ref{lemma-A}, the proof of Lem.~\ref{lemma-B} is also constructive.

%% file: sec_3-4.tex
\subsection{Algorithm to find roots of $E(y)$}

We have all the ingredients to give the algorithm for finding roots of $E(y)$ modulo ideal $\gen{p^4,\varphi^{4a}}$ of $\Z[x]$.

%Now we present our algorithm to find all the roots of $E(y) \bmod \langle p^4, \varphi^{4a}\rangle$. 

\noindent
\textbf{Input:} A polynomial $E(y) \in R[y]$ defined as $E(y) := f(x) (\varphi^{3a} + \varphi^{2a} (py) +\varphi^a (py)^2 + (py)^3)$. 

\noindent
\textbf{Output:} A set $Z \subseteq R_0$ and a \emph{bad} set $Z' \subseteq R_0$, such that, for each $y_0 \in Z-Z'$, there are (efficiently computable) $y_1 \in R_0$ (Theorem \ref{thm-algo1}) satisfying
$E(y_0+p y_1) \,\equiv\, 0\ \bmod \gen{p^4,\varphi^{4a}}$. These are exactly the roots of $E$.  

Also, both sets $Z$ and $Z'$ can be described by $O(a)$ many representatives (Theorem \ref{thm-algo1}). 
Hence, a $y_0\in Z-Z'$ can be picked efficiently. % This sentence needs to be improved.

\begin{breakablealgorithm}
\caption{Finding all roots of $E(y)$ in $R$}
\label{algo1}
\begin{algorithmic}[1]
%\Procedure{Root-find-multivariate}{$E(y)$,$R$}
\parState{Given $E(y_0+p y_1)$, using Lemma \ref{lemma-A}, get the set $S_0$ of all representative pairs $(a_0,i_0)$, where $a_0 \in R_0$ and $i_0\in \N$, such that 
$p^3| E((a_0+\varphi^{i_0} y_0)+p y_1) \bmod \gen{p^4,\varphi^{4a}}$.}
\parState{Initialize sets $Z=\{\}$ and $Z'=\{\}$; seen as subsets of $R_0$.}
\For{each $(a_0,i_0) \in S_0$}
\parState{Substitute $y_0 \mapsto a_0+\varphi^{i_0}y_0$, let $E'(y_0,y_1)=E_1(y_0)+E_2(y_0)y_1\bmod \gen{p,\varphi^{4a}}$ be the polynomial obtained from Lemma
\ref{lemma-A}.}
\parState{\textbf{If} $E_2(y_0)\not\equiv 0 \bmod \langle p,\varphi\rangle$ \textbf{then} find (at most one) $\theta \in R_0/\langle \varphi\rangle$ such that $E_2(\theta)\equiv 0 \bmod \langle p,\varphi\rangle$. Update $Z\leftarrow Z\cup (a_0+\varphi^{i_0} *)$ and $Z'\leftarrow Z'\cup (a_0+\varphi^{i_0}(\theta+\varphi *))$.}
\For{each possible valuation $r \in[4a]$}
\parState{Initialize sets $Z_r=\{\}$ and $Z_r^{'}=\{\}$.}
\parState{Call \Call{Root-Find}{$E_1$, $\varphi^r$} to get a set $S_1$ of representative pairs $(a_1,i_1)$ where $a_1 \in R_0$ and $i_1 \in \N$ such that $E_1(a_1+\varphi^{i_1} y_0)\equiv 0 \bmod \langle p,\varphi^{r}\rangle$.}
\For{each $(a_1,i_1) \in S_1$}
\parState{Analogously consider $E_2'(y_0) := E_2(a_1+\varphi^{i_1} y_0) \bmod \langle p,\varphi^{4a}\rangle$.}
\parState{Call \Call{Root-Find}{$E_2'$, $\varphi^r$} to get a representative pair $(a_2,i_2)$ ($\because E_2'$ is linear), where $a_2 \in R_0$ and $i_2 \in \N$ such that $E_2'(a_2+\varphi^{i_2} y_0)\equiv 0 \bmod \langle p,\varphi^{r}\rangle$.}
\If{$r=4a$}
\parState{Update $Z_r \leftarrow Z_r \cup (a_1+\varphi^{i_1}( a_2+\varphi^{i_2}*))$ and $Z_r^{'}\leftarrow Z_r^{'}\cup \{\}$.}
\ElsIf{$E_2'(a_2+\varphi^{i_2} y_0)\not \equiv 0 \bmod \langle p,\varphi^{r+1}\rangle$}
\parState{Get a $\theta \in R_0/\langle \varphi\rangle$ (Lemma \ref{lemma-B}), if it exists, such that $E_2'(a_2+\varphi^{i_2}(\theta+\varphi y_0))\equiv 0 \bmod \langle p,\varphi^{r+1}\rangle$. Update $Z_r^{'} \leftarrow Z_r' \cup (a_1+\varphi^{i_1} (a_2+\varphi^{i_2}(\theta+\varphi *)))$.}
\parState{Update $Z_r \leftarrow Z_r \cup (a_1+\varphi^{i_1}( a_2+\varphi^{i_2}*))$.}
\EndIf
\EndFor
\parState{Update $Z \leftarrow Z \cup (a_0+\varphi^{i_0} Z_r) $ and $Z' \leftarrow Z' \cup (a_0+\varphi^{i_0} Z_r^{'})$.}
\EndFor
\EndFor
\parState{Return $Z$ and $Z'$.}
%\EndProcedure
\end{algorithmic}
\end{breakablealgorithm}

We prove the correctness of Algorithm \ref{algo1} in the following theorem.

\begin{theorem}
\label{thm-algo1}
The output of Algorithm \ref{algo1} (set $Z-Z'$) contains exactly those $y_0 \in R_0$ for which there exist some $y_1\in R_0$, such that, $y=y_0+py_1$ is a root of $E(y)$ in $R$. We can easily compute the set of $y_1$ corresponding to a given $y_0\in Z-Z'$ in poly($\deg f, \log p$) time. 

Thus, we efficiently describe (\& exactly count) the roots $y= y_0+p y_1+p^2 y_2$ in $R$ of $E(y)$, where $y_0,y_1 \in R_0$ are as above and $y_2$ can assume any value from $R$.
\end{theorem}
\begin{proof}
The algorithm intends to output roots $y$ of equation $E(y)\equiv f(x) (\varphi^{3a}+\varphi^{2a}(py)+\varphi^{a}(py)^2+(py)^3)\equiv 0 \bmod \langle p^4,\varphi^{4a}\rangle$, where $y=y_0+p y_1+p^2 y_2$ with $y_0,y_1 \in R_0$ and $y_2\in R$. 
From Lemma~\ref{lemma-1}, $y_2$ can be kept as $*$, and is independent of $y_0$ and $y_1$. %Note that $y_3$ is irrelevant because $y$ is always multiplied by $p$ in $E(y)$.

Using Lemma~\ref{lemma-A}, Algorithm \ref{algo1} partially fixes $y_0$ from the set $S_0$ and reduces the problem to finding roots of an $E'(y_0,y_1) \bmod \gen{p,\varphi^{4a}}$.
In other words, if we can find all roots $(y_0,y_1)$ of $E'(y_0,y_1) \bmod \gen{p,\varphi^{4a}}$, then we can find (and count) all roots of $E(y) \bmod \gen{p^4,\varphi^{4a}}$. 
This is accomplished by Step $1$. From Lemma~\ref{lemma-A}, $|S_0|\leq 2$, so loop at Step $3$ runs only for a constant number of times.

Using Lemma~\ref{lemma-A}, $E'(y_0,y_1)\equiv E_1(y_0)+E_2(y_0)y_1 \bmod \gen{p,\varphi^{4a}}$ for a cubic polynomial $E_1(y_0)\in R_0[y_0]$ and a linear polynomial $E_2(y_0)\in R_0[y_0]$. 
%Though this is a bivariate equation modulo $\langle p,\varphi^{4a}\rangle$, we see that it is a special kind of bivariate equation. 

We find all solutions of $E'(y_0,y_1)$ by going over all possible valuations of $E_2(y_0)$ with respect to $\varphi$. The case of valuation $0$ is handled in Step 5 and valuation $4a$ is handled in Step 12. 
For the remaining valuations $r \in [4a-1]$, Lemma~\ref{lemma-2} shows that it is enough to find $(z_0,z_1) \in R_0\times R_0$ such that $\varphi^r | E_1(z_0)$ and $\varphi^r || E_2(z_0)$.

%Lemma~\ref{lemma-2} simplifies our task by stating the equivalence, 
%$(z_0,z_1)\in (R_0)^2$ is a root of $E'(y_0,y_1) \bmod \langle p,\varphi^{4a}\rangle$ iff $val_{\varphi}(E_1(z_0))\geq val_{\varphi}(E_2(z_0))$ iff either $val_{\varphi}(E_2(z_0))=0$ or $\exists r\in \{1,\ldots,4a-1\}$ such that $\varphi^r | E_1(z_0)$ and $\varphi^r || E_2(z_0)$.
%

%This gives us a different perspective to look at the problem from the angle of valuations $r$ (which are only constant in number). So at Step $5$ the algorithm finds those $y_0$s for which $val_{\varphi}(E_2(y_0))=0$. 
	
Notice that the number of valuations is bounded by $4a = O(\deg f)$. At Step $6$, the algorithm guesses the valuation $r$ of $E_2(y_0)\in R_0[y_0]$ and subsequent conputation finds all representative roots $b+\varphi^i *$ efficiently (using Theorem~\ref{thm-random-root-find}), such that,
$$ E_1(b+\varphi^i y_0)\,\equiv\, E_2(b+\varphi^i y_0) \,\equiv\, 0\, \bmod \gen{p,\varphi^{r}} \,.$$
The representative root $b+\varphi^i *$ is denoted by $a_1 + \varphi^{i_1}(a_2+\varphi^{i_2} *)$ in Steps 13 \& 16 of Algorithm \ref{algo1}.

Finally, we need to filter out those $y_0$'s for which $E_2(b+\varphi^i y_0)\equiv 0 \bmod \gen{p,\varphi^{r+1}}$. 
This can be done efficiently using Lemma~\ref{lemma-B}, where we get a unique $\theta \in R_0/\gen{\varphi}$ for which, 
\[ E_2(b+\varphi^i (\theta+\varphi y_0))\equiv 0 \bmod \gen{p,\varphi^{r+1}}. \]
	
We store partial roots in two sets $Z_r$ and $Z_r'$, where $Z_r'$ contains the bad values filtered out by Lemma~\ref{lemma-B} as ${b+\varphi^{i}(\theta+\varphi *)}$ and 
$Z_r$ contains all possible roots $b+\varphi^i *$. 
So, the set $Z_r - Z_r'$ contains exactly those elements $z_0$ for which there exists  $z_1\in R_0$, such that, the pair $(z_0,z_1)$ is a root of $E'(y_0,y_1) \bmod \gen{p,\varphi^{4a}}$.

Note that size of each set $S_1$ obtained at Step $9$ is bounded by three using Theorem~\ref{thm-random-root-find} ($E_1$ is at most a cubic in $y_0$). 
Again using Theorem~\ref{thm-random-root-find}, we get at most one pair $(a_2,i_2)$ at Step $11$ for some $a_2\in R_0$ and $i_2\in \N$ ($E_2'$ is linear in $y_0$).

Now, for a fixed $z_0 \in Z_r-Z_r'$ we can calculate all $z_1$'s by the equation 
\[ z_1\equiv \tilde{z}_1 := -(C(y_0)/L(y_0)) \bmod \gen{p, \varphi^{4a-r}} .\]
Here $C(y_0):= E_1(z_0)/\varphi^r \bmod \gen{p, \varphi^{4a-r}}$ and $L(y_0):= E_2(z_0)/\varphi^r \bmod \gen{p, \varphi^{4a-r}}$. 
So, $z_1\in R_0$ comes from the set $z_1\in \tilde{z}_1+\varphi^{4a-r} *$. This can be done efficiently in poly($\deg f, \log p$) time.

Finally, sets $Z= a_0+\varphi^{i_0}Z_r$ and $Z'=a_0+\varphi^{i_0}Z_r'$, for $(a_0,i_0)\in S_0$ and corresponding valid $r\in \{0,\ldots,4a-1\}$, returned by Algorithm \ref{algo1}, describe the $y_0$ for the roots of $E(y_0+p y_1)\bmod \gen{p^4, \varphi^{4a}}$. 
The number of representatives in each of these sets is $O(a)$, since $|S_0|\leq 2$ and sizes of $Z_r$ and $Z_r'$ are only constant.

Since we can efficiently describe these $y_0$'s and corresponding $y_1$'s, and we know their precision, we can count all roots $y=y_0+p y_1+p^2 * \subseteq R$ of 
$E(y) \bmod \gen{p^4,\varphi^{4a}}$.
\end{proof}

%\noindent
%\textbf{Note :}\label{note-1} Each root $y\in R_0$ obtained through Theorem~\ref{thm-algo1} has degree at most $4a \deg(\varphi)$ in the variable $x$. So, we only count (and find) roots with this bound on degree. 
%This bound is necessary, otherwise there are infinitely many roots. 
%% We should also mention that any * after $\varphi^{4a}$ will work.

\begin{remark}[Root finding for $k=3$ and $k=2$]
Algorithm \ref{algo1} can as well be used when $k\in \{2,3\}$ (the algorithm simplifies considerably).
 
For $k=3$, by Lemma~\ref{lemma-1}, the only relevant coordinate is $y_0$. Moreover, we can directly call algorithm \Call{Root-Find}{} to find all roots of $E(y)/p^2$. 

For $k=2$, using Lemma~\ref{lemma-1} again, we see that there are only two possibilities: $y_0=*$, or there is no solution. This can be determined by testing whether $E(y)/p^2 \bmod \langle \varphi^{2a}\rangle$ exists.
\end{remark}

%% file: sec_3-5.tex
\subsection{Wrapping up Theorems~\ref{thm1} \& \ref{thm2} }   % same heading for the section

%We have all the ingredients available to prove our main results.

\begin{proof}[Proof of Theorem~\ref{thm1}]
We prove that given a general univariate $f(x)\in \Z[x]$ and a prime $p$, a non-trivial factor of $f(x)$ modulo $p^4$ can be obtained in randomized poly$(\deg f,\log p)$ time (or the irreducibility of $f(x) \bmod p^4$ gets certified).

If $f(x) \equiv f_1(x) f_2(x) \bmod p$, where $f_1(x), f_2(x) \in \F_p[x]$ are two coprime polynomials, then we can efficiently lift this factorization 
to the ring $(\Z/\gen{p^4})[x]$, using Hensel lemma (Lemma~\ref{lemma-hensel}), to get non-trivial factors of $f(x) \bmod p^4$. 

For the remaining case, $f(x)\equiv \varphi^e \bmod p$ for an irreducible polynomial $\varphi(x)$ modulo $p$. The question of factoring $f \bmod p^4$ then reduces to root finding of a polynomial $E(y) \bmod \gen{p^4,\varphi^{4a}}$ by Reduction theorem (Theorem~\ref{thm-reduction}). Using Theorem~\ref{thm-algo1}, we get all such roots and hence a non-trivial factor of $f(x) \bmod p^4$ is found. If there are no roots $y \in R$ of $E(y)$, for all $a\leq e/2$, then the polynomial $f$ is irreducible (by symmetry, if there is a factor for $a > e/2$ then there is a factor for $a \leq e/2$).
\end{proof}

% This can go inside the proof.
\begin{remark}
As discussed before, the above proof applies to factorization modulo $p^3$ and $p^2$ as well (by considering the generality of Theorems~\ref{thm-reduction} \& \ref{thm-algo1}). Hence, Theorem~\ref{thm1} also solves the open question of factoring $f$ modulo $p^3$. In fact, in Appendix \ref{appendix-mod-p^3} we observe that our efficient algorithm outputs {\em all} the factors of $f\bmod p^3$ in a compact way.
\end{remark}

\begin{proof}[Proof of Theorem \ref{thm2}]
We will prove the theorem for $k=4$, case of $k<4$ is similar.

We are given a univariate $f(x) \in \mathbb{Z}[x]$ of degree $d$ and a prime $p$, such that, $f(x) \bmod p$ is a power of an irreducible polynomial $\varphi(x)$. 
So, $f(x)$ is of the form $\varphi(x)^e + p h(x) \bmod p^4$, for an integer $e\in \N$ and a polynomial $h(x) \in (\Z/\gen{p^4})[x]$ of degree $\leq d$ (also, $\deg\varphi^e\le d$). 
%so that $f(x) \equiv \varphi^e \bmod p$. 
By unique factorization over the ring $\F_p[x]$, if $\tilde{g}(x)$ is a factor of $f(x)\bmod p$ then, $\tilde{g}(x)\equiv \tilde{v}\varphi(x)^a \bmod p$ for a unit $\tilde{v} \in \F_p$. 

First, we show that it is enough to find all the lifts of $\tilde{g}(x)$, such that, $\tilde{g} (x) \equiv \varphi(x)^a \bmod p$ for an $a\leq e$. 
If $\tilde{g}(x) \equiv \tilde{v} \varphi(x)^a \bmod p$, then any lift has the form $g(x) \equiv v(x)(\varphi^a-p y) \bmod p^4$ for a unit 
$v(x)\in (\tilde{v}+p *) \subseteq (\Z/\gen{p^4})[x]$. 
Any such $g(x)$ maps uniquely to a $g_1(x) := \tilde{v}^{-1} g(x) \bmod p^4$, which is a lift of $\varphi(x)^a \bmod p$. So, it is enough to find all the lifts of $\varphi(x)^a \bmod p$.

%First we show that any factor $\tilde{g}(x) \in \F_p[x]$ of $f(x) \bmod p$ satisfies $\tilde{g} (x) \ \equiv\ \varphi(x)^a \bmod p$ for an $a\leq e$.
%Let $g(x) \ \equiv\ v(x)(\varphi^a-p y) \bmod p^k$ be a lift of $\tilde{g}(x)$ for a unit $v(x)\in (\tilde{v}+p *) \subseteq (\Z/\langle p^k\rangle)[x]$. 
%Take $g_1(x)\equiv \tilde{v}g(x) \bmod p^k$. Since $g_1(x)\equiv \varphi^a \bmod p$, so it is a lift of $\varphi^a \bmod p$. So all the lifts of $\tilde{g}\bmod p$ are the lifts of $\varphi^a \bmod p$ multiplied by constant $\tilde{v}\in \F_p$.

We know that any lift $g(x) \in (\Z/\gen{p^4})[x]$ of $\tilde{g}(x)$, which is a factor of $f(x)$, must be of the form $\varphi(x)^a - py(x) \bmod p^4$ for a polynomial $y(x) \in (\Z/\gen{p^4})[x]$.
By Reduction theorem (Theorem~\ref{thm-reduction}), we know that finding such a factor is equivalent to solving for $y$ in the equation $E(y)\equiv 0 \bmod \gen{p^4,\varphi^{4a}}$. 
By Theorem \ref{thm-algo1}, we can find all such roots $y$ in randomized poly($d, \log p$) time, for $a \leq e/2$. 

If $a > e/2$ then we replace $a$ by $b := e-a$, as $b \leq e/2$, and solve the equation $E(y)\equiv 0 \bmod \gen{p^4,\varphi^{4b}}$ using Theorem \ref{thm-algo1}. 
%to get all the roots $y$ but 
This time the factor corresponding to $y$ will be,
$g(x) \ \equiv\ f/(\varphi^{b}-py) \ \equiv\ E(y)/\varphi^{4b} \ \bmod p^4$,
from Reduction theorem (Theorem~\ref{thm-reduction}).

%Now, let us lift $\tilde{g}$ to get some $g(x) \in (\Z/\langle p^4\rangle)[x]$ such that $g(x)| f(x) \bmod p^4$. Such a $g(x)$ must be of the form $(\varphi^{a}-p y) \bmod p^4$ for some $y \in (\Z/\langle p^k\rangle)[x]$. 
%By Reduction theorem (Theorem \ref{thm-reduction}), we know that that finding such a factor is equivalent to solving for $y$ the equation $E(y)\equiv 0 \bmod \langle p^4,\varphi^{4 a}\rangle$. By Theorem \ref{thm-algo1}, we can find all such roots $y$ satisfying the equation and so all such factors $g(y) \equiv (\varphi^{a}-p y) \bmod p^4$ in randomized poly($d, \log p$)-time but that requires $a \leq e/2$. So if $a \leq e/2$ we are done, if not then we fix some $b=e-a$ so that $b \leq e/2$ and solve the equation $E(y)\equiv 0 \bmod \langle p^4,\varphi^{4 b}\rangle$ using Theorem \ref{thm-algo1} to get all the roots $y$ but this time the factor corresponding to a $y$ will be $g(x) \equiv (f(x)/\varphi^{4 b}) (\varphi^{3 b}+\varphi^{2 b} (p y)+\varphi^{b} (p y)^2+(p y)^3)  \bmod p^4$ (by Reduction theorem \ref{thm-reduction}).
%

The number of lifts of $\tilde{g}(x)$ which divide $f\bmod p^4$ is the count of $y$'s that appear above. This is efficiently computable.
%is the count on the number of all the lifts using Reduction theorem (Theorem \ref{thm-reduction}). 
%If there is no root $y$ satisfying the equation $E(y)$ for any $a\in [e-1]$ 
%then using Reduction theorem (Theorem \ref{thm-reduction}) there is no lift $g(x)$ of $\tilde{g}$ for any $a\in [e-1]$ which divides $f(x) \bmod p^4$ and so we say that then $f \bmod p^4$ is irreducible.
\end{proof}

% Do we need this remark?
%\begin{remark}
%Again, the proof of Theorem \ref{thm2} is applicable to the case $k=3$ and $k=2$, since Theorem \ref{thm-reduction} and Theorem \ref{thm-algo1} works in these cases as mentioned before.
%\end{remark}

%% file: appendix.tex
\section{Preliminaries}
\label{appendix-pre}

The following theorem by Cantor-Zassenhaus \cite{cantor1981new} efficiently finds all the roots of a given univariate polynomial over a finite field. 
%This theorem will be referred by CZ in this paper.

\begin{theorem}[Cantor-Zassenhaus]
\label{thm-CZ}
Given a univariate degree $d$ polynomial $f(x)$ over a given finite field $\mathbb{F}_q$, we can find all the irreducible factors of $f(x)$ in $\F_q[x]$ in randomized poly($d,\log{q}$) time.
\end{theorem}

Currently, it is a big open question to derandomize this algorithm.

Below we state a lemma, originally due to Hensel \cite{Hensel1918}, for $\mathcal{I}$-adic lifting of {\em coprime} factorization for a given univariate polynomial. Over the years, it has acquired many forms in different texts; the version being presented here is due to Zassenhaus \cite{zassenhaus1969hensel}.

\begin{lemma}[Hensel lemma \& lift \cite{Hensel1918}]
\label{lemma-hensel}
Let $R$ be a commutative ring with unity, and let $\mathcal{I}\subseteq R$ be an ideal. Given a polynomial $f(x) \in R[x]$, let 
$g(x),h(x),u(x),v(x) \in R[x]$ be polynomials, such that, $f(x) = g(x) h(x) \bmod \mathcal{I}$ and $g(x)u(x)+h(x)v(x)=1 \bmod \mathcal{I}$.
%Let $R[x]$ be polynomials over a commutative ring, with unity, $R$ and let $\mathcal{I}\subseteq R$ be an ideal of ring $R$. Given a polynomial $f(x) \in R[x]$ which factorizes as 
%$f=gh \bmod \mathcal{I} $, 
%for some $g(x),h(x) \in R[x]$ such that $gu+hv=1 \bmod  \mathcal{I}$ for some $u,v \in R[x]$. 

Then, for any $l\in \mathbb{N}$, we can efficiently compute $g^*,h^*,u^*,v^* \in R[x]$ such that
\[f \,=\, g^* h^* \,\bmod \mathcal{I}^{l} \qquad\text{ (called {\em lift of the factorization}) }\]
where $g^*=g \bmod \mathcal{I}$, $h^*=h \bmod \mathcal{I}$ and $g^*u^*+h^*v^*=1 \bmod \mathcal{I}^{l}$. 

Moreover, $g^*$ and $h^*$ are unique upto multiplication by a unit.
\end{lemma}

\section{Root finding modulo $\varphi(x)^i$}
\label{appendix-rootfind}

Let us denote the ring $\F_p[x]/\langle \varphi^i\rangle$ by $R_0$ (for an irreducible $\varphi(x)\bmod p$). In this section, we give an algorithm to find all the roots $y$ of a polynomial $g(y) \in R_0[y]$ in the ring $R_0$. 
The algorithm was originally discovered by \cite[Cor.24]{berthomieu2013polynomial} to find roots in $\Z/\langle p^i\rangle$, we adapt it here to find roots in $R_0$.

%We first present the algorithm to find roots of $g(y)$ in $R_0$. 
Note that $R_0/\langle \varphi^j\rangle = \F_p[x]/\langle \varphi^j\rangle$, for $j\leq i$, and $R_0/\langle \varphi\rangle =: \F_q$ is the finite field of size $q:=p^{\deg(\varphi\bmod p)}$. 
The structure of a root $y$ of $g(y)$ in $R_0$ is 
\[y \,=\, y_0+\varphi y_1+\varphi^2 y_2+\ldots+ \varphi^{i-1} y_{i-1},\]
where $y\in R_0$ and each $y_j\in \F_q$ for all $j\in \{0,\ldots,i-1\}$. Also, recall the notation of $*$ (given in Section \ref{sec-pre}) and representative roots (in Section \ref{fact-root}).

The {\bf output} of this algorithm is simply a set of at most ($\deg g$) many representative roots of $g(y)$. This bound of $\deg g$ is a curious by-product of the algorithm.

\begin{breakablealgorithm}
 \caption{Root-finding in  ring $R_0$}\label{algo2}
\begin{algorithmic}[1]
\Procedure{Root-find}{$g(y)$, $\varphi^i$}
\parState{\textbf{If} $g(y)\equiv 0$ in $R_0/\langle \varphi^i\rangle$ \Return{ $*$} (every element is a root).}
\parState{Let $g(y)\equiv \varphi^\alpha \tilde{g}(y)$ in $R_0/\langle \varphi^i\rangle$, for the unique integer $0\leq\alpha< i$ and the polynomial $\tilde{g}(y)\in R_0/\langle\varphi^{i-\alpha}\rangle[y]$, s.t.,~$\tilde{g}(y)\not \equiv 0$ in $R_0/\langle \varphi\rangle$ and $\deg(\tilde{g})\leq \deg(g)$.}

\parState{Using Cantor-Zassenhaus algorithm find all the roots of $\tilde{g}(y)$ in $R_0/\langle \varphi \rangle$.}
\parState{\textbf{If} $\tilde{g}(y)$ has no root in $R_0/\langle \varphi \rangle$ then \Return{ $\{ \}$}. (Dead-end)}
\parState{Initialize $S = \{\}$. }%(for storing all the roots of $g(y)$ in $R_0/\langle \varphi^i\rangle$)}
\For{each root $a$ of $\tilde{g}(y)$ in $R_0/\langle \varphi\rangle$}
\parState{Define $g_a(y):=\tilde{g}(a+\varphi y)$.}
\parState{$S' \leftarrow$\Call{Root-find}{$g_a(y)$, $\varphi^{i-\alpha}$}.}
\parState{$S\leftarrow S\cup (a+\varphi S')$.}
\EndFor
\parState{\Return{$S$}.}
\EndProcedure
\end{algorithmic}
\end{breakablealgorithm}

Note that in Step 9 we ensure: $\varphi | g_a(y)$. So, in every other recursive call to \Call{Root-find}{} the second argument reduces by at least one. 
The key reason why $|S|\le \deg g$ holds: The number of representative roots of $g_a(y)$ are upper bounded by the multiplicity of the root $a$ of $\tilde{g}(y)$.

\section{Finding all the factors modulo $p^3$}
\label{appendix-mod-p^3}

We will give a method to efficiently get and count all the distinct factors of $f \bmod p^3$, where $f(x)\in \Z[x]$ is a univariate polynomial of degree $d$.

\begin{theorem}\label{thm-p^3}
Given $f(x)\in \Z[x]$, a univariate polynomial of degree $d$ and a prime $p\in \N$, we give (\& count) all the distinct factors of $f \bmod p^3$ of degree at most $d$ in randomized poly($d,\log p$) time.
\end{theorem}

\noindent
\textbf{Note:}\label{note-p^3}
We will not distinguish two factors if they are same up to multiplication by a unit. We will only find monic (leading coefficient $1$) factors of $f(x)\bmod p^3$ and assume that $f$ is monic. %To give all distinct factors not distinguished by a unit it is a fair assumption to make.

\begin{proof}[Proof of Theorem \ref{thm-p^3}]
By Theorem \ref{thm-CZ} and Lemma \ref{lemma-hensel} we write: 
\[f(x) \,\equiv\, \prod_{i=1}^{n} f_i(x) \,\equiv\, \prod_{i=1}^{n}(\varphi_i^{e_i}+p h_i)\, \bmod p^3\]
where $f_i(x)\equiv (\varphi_i^{e_i}+p h_i) \bmod p^3$ with $\varphi_{i} \bmod p^3$ being monic and irreducible mod $p$, $e_i\in \N$, and $h_i(x)\bmod p^3$ of degree  $< e_i \deg(\varphi_{i})$, for all $i\in [n]$.

Thus, wlog, consider the case of $f\equiv \varphi^e+ p h$.% for some $\varphi \in \Z[x]$ which is monic mod $p^3$ (leading coeff. 1) and irreducible mod $p$, $e\in \N$ and $h(x)\in \Z/\langle p^3\rangle[x]$ of degree less than $e \deg(\varphi)$.

By Reduction theorem (Theorem \ref{thm-reduction}) finding factors of the form $\varphi^a- p y \bmod p^3$ of $f\equiv \varphi^e+ p h \bmod p^3$, for $a\leq e/2$, is equivalent to finding all the roots of the equation
\[ E(y) \,\equiv\, f(x) (\varphi^{2a}+\varphi^{a} (p y)+ (p y)^2)\,\equiv\, 0\,  \bmod \langle p^3,\varphi^{3a}\rangle . \]

Consider $R:= \Z[x]/\langle p^3,\varphi^{3a}\rangle$ and $R_0:= \Z[x]/\langle p,\varphi^{3a}\rangle$, analogous to those in Section \ref{sec-pre}.

Using Lemma \ref{lemma-1}, we know that all solutions of the equation $E(y) \equiv 0 \bmod  \langle p^3,\varphi^{3a}\rangle$ will be of the form $y=y_0+p* \in R$, for a $y_0\in R_0$.
On simplifying this equation we get $E(y) \,\equiv\, p h \varphi^{2a}+(p^2 h \varphi^a) y_0 + (p^2 \varphi^e) y_0^2 \,\equiv\, 0\, \bmod \langle p^3,\varphi^{3a}\rangle$. 

Reducing this equation mod $\langle p^2,\varphi^{3a}\rangle$, we get that $h \equiv 0 \bmod \gen{p, \varphi^a}$ is a necessary condition for a root $y_0$ to exist. So, we get  
$$E(y) \,\equiv\, p^2 g_2 \varphi^{2a}+(p^2 g_1 \varphi^{2a}) y_0 + (p^2 \varphi^e) y_0^2 \,\equiv\, 0\, \bmod \langle p^3,\varphi^{3a}\rangle, $$
where $h:= \varphi^a g_1+p g_2$ for unique $g_1, g_2\in \F_p[x]$.% and a unique $g_1 \in \Z[x]/\langle p^2\rangle$.

This equation is already divisible by $p^2$ as well as $\varphi^{2a}$ and so using  Claim \ref{claim-2} we get that, finding factors of the form $\varphi^a- p y \bmod p^3$ of $f\equiv \varphi^e+ p h \bmod p^3$, for $a\leq e/2$, is equivalent to finding all the roots of the equation
\[ g_2 +g_1 y_0 +  \varphi^{e-2a} y_0^2 \,\equiv\, 0\, \bmod \langle p,\varphi^{a}\rangle \,.\]

We find all the roots of this equation using one call to \Call{Root-find}{} in randomized poly($d, \log p$) time. Note that any output root $u_0$ lives in $\F_p[x]/\langle \varphi^a\rangle$ and so its degree in $x$ is $< a\deg(\varphi)$. This yields {\em monic} factors of $f\bmod p^3$ (with $0\le a\le e/2$).

For $e\ge a> e/2$, we can replace $a$ by $b:=e-a$ in the above steps. Once we get a factor $\varphi^b- p y \bmod p^3$, we output the cofactor $f/(\varphi^b- p y)$ (which remains monic).

Counting these factors can be easily done in poly-time.

\smallskip
%Now we assume $f\bmod p^3$ as usual using the description above we efficiently find all the monic factors of each $f_i\bmod p^3$. Now any factor of $f(x)\bmod p^3$ is product of factors of some $f_i$s so we efficiently get all the distinct factors of $f\bmod p^3$.

In the general case, if $N_i$ is the number of factors of $f_i\bmod p^3$ then, $\prod_{i=1}^{n}N_i$ is the count on the number of distinct monic factors of $f\bmod p^3$.
\end{proof}

\section{Barriers to extension modulo $p^5$}
%\section{Extension to mod $p^5$ and future work}
\label{sec-comparison}

The reader may wonder about polynomial factoring when $k$ is greater than $4$.
In this section we will discuss the issues in applying our techniques to factor $f(x) \bmod p^5$.

Given $f(x)\equiv \varphi^e \bmod p$, finding one of its factor $\varphi^a-p y \bmod p^5$, for $a\leq e/2$ and $y\in (\Z/\langle p^5\rangle) [x]$, is reduced to solving the equation
\begin{equation}\label{eq3}
E(y) \ :=\ f(x)( \varphi^{4a}+ \varphi^{3a}(py)+ \varphi^{2a}(py)^2+ \varphi^a (py)^3 +(py)^4) \ \equiv\  0\ \bmod \langle p^5,\varphi^{5a} \rangle
\end{equation}

By Lemma \ref{lemma-1}, the roots of $E(y) \bmod \langle p^5,\varphi^{5a}\rangle$ are of the form $y= y_0+p y_1+ p^2 y_2+ p^3 *$ in $R$, where $y_0,y_1,y_2\in R_0$ need to be found.

{\bf First issue.}
The first hurdle comes when we try to reduce root-finding modulo the bi-generated ideal $\langle p^5,\varphi^{5a}\rangle \subseteq \Z[x]$ to root-finding modulo the principal ideal $\langle \varphi^{5a}\rangle \subseteq \F_p[x]$. In the case $k=4$, Lemma \ref{lemma-A} guarantees that we need to solve at most two related equations of the form $E'(y_0,y_1)\equiv 0 \bmod \langle p, \varphi^{4a}\rangle$ to find exactly the roots of $E(y) \bmod \langle p^4,\varphi^{4a}\rangle$. Below, for $k=5$, we show that we have exponentially many candidates for $E'(y_0,y_1,y_2) \in R_0[y_0,y_1,y_2]$ and it is not clear if there is any compact efficient representation for them.

Putting $y= y_0+p y_1+ p^2 y_2$ in Eqn.~\ref{eq3} we get,
\begin{equation}\label{eq4}
E(y) \ =:\ E_1(y_0) + E_2(y_0)y_1 + E_3(y_0)y_2 + (f\varphi^{2a}p^4) y_1^2 \ \equiv\ 0\ \bmod \langle p^5,\varphi^{5a}\rangle,
\end{equation}
where $E_1(y_0):= f\varphi^{4a}+ f\varphi^{3a} p y_0+ f\varphi^{2a} p^2 y_0^2+ f\varphi^a p^3 y_0^3 + f p^4 y_0^4$ is a quartic in $R[y_0]$, $E_2(y_0):= f\varphi^{3a} p^2+ f \varphi^{2a} 2 p^3 y_0+ f \varphi^a 3 p^4 y_0^2$ is a quadratic in $R[y_0]$ and $E_3(y_0):= f\varphi^{3a} p^3 + f\varphi^{2a} 2 p^4 y_0$ is linear  in $R[y_0]$.

%We see that the problem is recursive in nature i.e., solving for higher prime powers $p^k$ is at least as hard as solving for some lower prime power $p^l$, $l< k$. 
To divide Eqn.~\ref{eq4} by $p^3$, we go mod $\langle p^3, \varphi^{5a}\rangle$ obtaining 
\[E(y) \,\equiv\, E_1(y_0)\equiv f\varphi^{4a}+ f\varphi^{3a} p y_0+ f\varphi^{2a} p^2 y_0^2 \,\equiv\, 0\, \bmod \langle p^3,\varphi^{5a}\rangle,\]
a univariate quadratic equation which requires the whole machinery used in the case $k=3$. 
We get this simplified equation since $E_3(y_0)\equiv 0\bmod \langle p^3,\varphi^{5a}\rangle$ and $E_2(y_0)\equiv f\varphi^{3a} p^2\equiv \varphi^{e-2a} \varphi^{2a+3a} p^2 \equiv 0\bmod \langle p^3,\varphi^{5a}\rangle$.

But, to really reduce Eqn.~\ref{eq4} to a system modulo the principal ideal $\langle \varphi^{5a}\rangle \subseteq \F_p[x]$, we need to divide it by $p^4$. 
So, we go mod $\langle p^4,\varphi^{5a}\rangle$:
\[E(y) \,\equiv\, E_1^{'}(y_0)+ E_2^{'}(y_0)y_1 \,\equiv\, 0\, \bmod \langle p^4,\varphi^{5a}\rangle\]
where $E_1^{'}(y_0)\equiv E_1(y_0)\bmod \langle p^4,\varphi^{5a}\rangle$ is a cubic in $R[y_0]$ and $E_2^{'}(y_0)\equiv E_2(y_0)\bmod \langle p^4,\varphi^{5a}\rangle$ is linear in $R[y_0]$. %Note that $E_3(y_0)\equiv f\varphi^{3a} p^3\equiv \varphi^{e}\varphi^{3a} p^3\equiv \varphi^{e-2a}\varphi^{2a+3a} p^3\equiv 0 \bmod \langle p^4,\varphi^{5a}\rangle$.
This requires us to solve a special bivariate equation which requires the machinery used in the case $k=4$.  

Now, the problem reduces to computing a solution pair $(y_0,y_1)\in (R_0)^2$ of this bivariate. 
We can apply the idea used in Algorithm \ref{algo1} to get all valid $y_0$ efficiently, but since $y_1$ is a function of $y_0$, we need to compute exponentially many $y_1$'s. 
So, there seem to be exponentially many candidates for $E'(y_0,y_1,y_2)$, that behaves like $E(y)/p^4$ and lives in $(\F_p[x]/\langle \varphi^{5a} \rangle)[y_0,y_1,y_2]$. 
At this point, we are forced to compute all these $E'$s, as we do not know which one will lead us to a solution of Eqn.~\ref{eq4}.

\smallskip {\bf Second issue.}
Even if we resolve the first issue and get a valid $E'$, we are left with a trivariate equation to be solved mod $\langle p,\varphi^{5a}\rangle$ 
(Eqn.~\ref{eq4} after shifting $y_0$ and $y_1$ then dividing by $p^4$). 
We could do this when $k$ was $4$, because we could easily write $y_1$ as a function of $y_0$. 
Though, it is unclear how to solve this trivariate now as the equation is {\em nonlinear} in both $y_0$ and $y_1$.

For $k>5$ the difficulty will only increase because of the recursive nature of Eqn.~\ref{eq3} with more and more number of unknowns (with higher degrees). %But the high valuation of the complicated terms may help in some way which is not clear as it stands even for $k=5$.